\documentclass[a4paper,11pt]{article}

\usepackage{titlefoot}
\usepackage{times}
\usepackage[english]{babel}

\usepackage{amsfonts,amsthm,amssymb,amsmath}
\usepackage[title,titletoc,toc]{appendix}
\usepackage[utf8]{inputenc}
\usepackage[affil-it]{authblk}
\usepackage{xcolor}
\usepackage{tikz}

\newtheorem{theorem}{Theorem}[section]

\newtheorem{lemma}[theorem]{Lemma}

\newtheorem{corollary}[theorem]{Corollary}
\newtheorem{definition}[theorem]{Definition}

\newcommand{\E}{\mathbb E}

\makeatletter
\newcommand*\bigcdot{\mathpalette\bigcdot@{.5}}
\newcommand*\bigcdot@[2]{\mathbin{\vcenter{\hbox{\scalebox{#2}{$\m@th#1\bullet$}}}}}
\makeatother
\newcommand{\is}{\bigcdot }

\def\comg#1{\left ( #1\right )\!^{p,\mathbb G}}

\def \Lbrack {[\![}
\def \Rbrack {]\!]}

\numberwithin{equation}{section}

\DeclareMathOperator*{\loc}{loc}

\DeclareMathOperator*{\Cov}{Cov}

\begin{document}
%\title{Mortality Risk-Minimization and Optional Martingale Representation}
\title{Mortality/longevity Risk-Minimization with or without securitization}

\author{Tahir Choulli\\ University of Alberta \and Catherine Daveloose and Mich\`ele Vanmaele\\Ghent University}

%\author[1]{Tahir Choulli}
%\author[2]{Catherine Daveloose}
%\author[2]{Mich\`ele Vanmaele}
%\affil[1]{\small{Department of Mathematical and Statistical Sciences, University of Alberta, 632 Central Academic Building,
%Edmonton, AB T6G 2G1, Canada, tchoulli@ualberta.ca}}
%\affil[2]{Department of Applied Mathematics, Computer Science, and Statistics, Ghent University,  Krijgslaan 281-S9, 9000 Gent, Belgium,  Catherine.Daveloose@UGent.be, Michele.Vanmaele@UGent.be}

% \date{}

\maketitle\unmarkedfntext{
This research is supported by NSERC (through grant NSERC RGPIN04987). The financial support from the Agency for Innovation by Science and Technology in Flanders (IWT, grant number 111262) is gratefully acknowledged by Catherine Daveloose.\\
Tahir Choulli is thankful to Mich\`ele Vanmaele and Department of Applied Mathematics, Computer Science, and Statistics (Ghent University),  for their invitation and  hospitality, where this work started.\\
Address correspondence to Mich\`ele Vanmaele, Department of Applied Mathematics, Computer Science and Statistics, Ghent University, Ghent, Belgium; e-mail: Michele.Vanmaele@UGent.be}

\begin{abstract}
This paper addresses the risk-minimization problem, with and without mortality securitization, {\it \`a la} F\"ollmer-Sondermann for a large class of equity-linked mortality contracts when no model for the death time is specified. This framework includes the situation where the correlation between the market model and the time of death is arbitrary general, and hence leads to the case of a market model where there are two levels of information.  The public information which is generated by the financial assets,
and a larger flow of information that contains additional knowledge about a death time of an insured. 
By enlarging the filtration, the death uncertainty and its entailed risk are fully considered without any mathematical restriction.
 Our key tool lies in our optional martingale representation that states  that any martingale in the large filtration stopped at the death time can be decomposed into precise orthogonal local martingales. This allows us to derive the dynamics of the value processes of the mortality/longevity securities used for the securitization, and to decompose any mortality/longevity liability into the sum of  orthogonal risks by means of a risk basis. The first main contribution of this paper resides in quantifying, as explicit as possible, the effect of mortality uncertainty on the risk-minimizing strategy by determining the optimal strategy in the enlarged  filtration in terms of strategies in the smaller filtration. Our second main contribution consists of finding risk-minimizing strategies with insurance securitization by investing in stocks and one (or more) mortality/longevity derivatives such as longevity bonds. This generalizes the existing literature on risk-minimization using mortality securitization in many directions.\\

\end{abstract}

\section{Introduction}
In this paper we manage the risk of a life insurance portfolio that faces two main types of risk: financial risk and mortality or longevity risk by designing quadratic hedging strategies {\it \`a la} F\"ollmer-Sondermann, introduced in \cite{follmersondermann86}, with and without mortality securization.
We consider  a financial setting consisting of an initial market, characterized by its flow of information $\mathbb F$ and its underlying traded assets $S$, and a random time, the death time $\tau$, that might not be observed through $\mathbb F$ when it occurs. The financial risk originates  from the investment in the risky assets, while the mortality risk follows from the uncertainty of the 
death time and can be split into a systematic and an unsystematic part, see, e.g.,  \cite{dhalmelchiormoller08, dhalmoller06} and the references therein. Longevity risk refers to the risk that the  reference population might, on average, live longer than anticipated. The unsystematic mortality risk, i.e., the risk corresponding to individual mortality rates, can be diversified by increasing the size of the portfolio while systematic mortality risk and longevity risk cannot be diversified away by pooling.   The market for mortality-linked instruments, the so-called
life market, to transfer such illiquid risks into financial markets as an alternative to the classical actuarial form of risk mitigation,  is in full development. 
 In \cite{blakeburrows01}, see also  \cite{barrieuetal09,blakeetal08b}  and the references therein, the authors were
the first to advocate the use of mortality-linked securities for hedging purposes. The first longevity bonds were sold in the late 1990s. The longevity derivatives market has since expanded to include forward contracts, options and swaps. A detailed account of this evolution can be found in \cite{tanetal2015, blakeetal2018}. The development of the life market entails questions about engineering of mortality-linked securities or derivatives as well as their pricing and finding their dynamics.
As the authors in \cite{jevticetal2017} state `The pricing of any mortality linked derivative security begins with the choice of a mortality model.', these prices obviously depend heavily on the chosen mortality model and the method used to price those securities. Since the Lee-Carter model introduced in \cite{leecarter92},  many  mortality models have been suggested. They can be classified into two main groups, depending on whether the obtained model was inspired from credit risk modelling or interest rate modelling. The first approach is based on the strong similarity between mortality and default and hence uses the arguments of credit risk theory, while the second approach follows the interest rate term structure approach such as in \cite{barbarin09}.
%While older models assume independence of financial and demographic risk factors, more recent models refrain from making this assumption due to the emerging body of empirical research suggesting a connection between long-run demographic trends and financial markets. 
%, exemplified in works by Ang and Maddaloni (2003), Favero et al. (2011) , Maurer (2011b) and Dacorogna and Cadena (2015),  
However model misspecification can have a significant impact on the performance of hedging strategies for mortality or longevity risk. Recently, in \cite{friedbergwebb07} (see also \cite{cairnsetal06, bauer10} for related discussion),  the authors use the CAPM and the CCAPM to price longevity bonds, and concluded that this pricing is not accurate with the reality and suggest that there might be a kind of `mortality premium puzzle' {\it \`a la } Mehra and Prescott \cite{Mehra1985}. While this mortality premium puzzle might exist, the `poor and/or bad' specification of the model for the mortality plays an important role in getting those wrong prices for longevity bonds.  In \cite{lietal2017} they propose a robust mean-variance hedging approach to deal with parameter uncertainty and model misspecification. 

Our aim is to position ourselves in a context \textit{without mortality specification}
and  derive the dynamics of the security's price process and design the risk-minimizing hedging strategies.\\

To further elaborate our main aim in this paper and its relation to the literature, we introduce some notations that are valid throughout the whole paper.
% we consider given  described mathematically by 
The tuplet $\left(\Omega, {\cal F}, {\mathbb F},S, P\right)$
represents mathematically the financial market model.  Herein, the filtered probability space $\left(\Omega,{\cal F},  {\mathbb  F}=({\cal F}_t)_{t\geq 0}, P\right)$ satisfies the usual condition (i.e., filtration is complete and  right continuous) with ${\cal F}_t\subset{\cal F}$, and $S$ is an $\mathbb F$-semimartingale representing the discounted price process of $d$ risky assets. The mortality is modelled with the death time of the insured, $\tau$, which is mathematically an arbitrary random time (i.e., a $[0,+\infty]$-valued random variable).  The flow of information generated by the public flow $\mathbb F$ enlarged by the random time will be denoted by $\mathbb G$, where the relationship between the three components $\mathbb F$, $\tau$ and $\mathbb G$ will be specified in the next section. \\

%Thus, our goals in this paper  can be summarized into two main objectives. 
 Up to our knowledge, apart from the recent paper \cite{choullidaveloose2018}, all the existing literature  about mortality and/or longevity assumes a specific model for mortality and derives the dynamics for longevity bonds prices accordingly. For an up-to-date extensive list of relevant papers, see \cite{lietal2017}. We follow the approach of \cite{choullidaveloose2018}.
Even though,  $\mathbb G$ is the progressive enlargement of $\mathbb F$ with $\tau$ as in credit risk theory, the death time is kept arbitrary general with no assumption at all. This translates into the fact that the survival probabilities over time constitute a general nonnegative supermartingale. To capture which risk to hedge under $\mathbb G$ we use the classification of \cite{choullidaveloose2018} where a $\mathbb G$-risk up to $\tau$ is expressed as a functional of pure financial risk (PFR), pure mortality risk (PM) and correlation risk (CR) intrinsic to the correlation between the financial market and the mortality. Mathematically, these risks can expressed as (local) martingales due to arbitrage theory. In \cite{choullidaveloose2018}, we elaborated an optional martingale representation for  martingales in the large filtration $\mathbb G$ stopped at the death time under no assumption of any kind. This representation states that any martingale in the large filtration $\mathbb G$ stopped at the death time can be decomposed into precise orthogonal local martingales. 
By means of  this optional martingale representation we derive the dynamics of the value processes of the mortality/longevity securities and decompose any mortality/longevity liability into the sum of  orthogonal risks by means of a risk basis. 
Then, \textit{our main objective} lies, when one considers the quadratic hedging {\it \`a la} F\"ollmer-Sondermann, in quantifying the functional $\Xi$ and $(\xi^{\rm pf}, \xi^{{\rm pm}}_1,\ldots, \xi^{{\rm pm}}_k,\xi^{\rm cr}_1,\dots,\xi^{\rm cr}_l)$  such that
\begin{equation}\label{StrategyDecomposition}
\xi^{\mathbb G}=\Xi\Bigl(\xi^{\rm pf}, \xi^{{\rm pm}}_1,\ldots, \xi^{{\rm pm}}_k,\xi^{\rm cr}_1,\ldots,\xi^{\rm cr}_l\Bigr).
\end{equation}
Here $\xi^{\mathbb G}$ is the optimal hedging strategy for the whole risk encountered under $\mathbb G$ on $\Lbrack0,\tau\Rbrack$, $\xi^{\rm pf}$ is the optimal hedging strategy for
the pure financial risk, $\xi^{\rm pm}_i$, $i=1,\ldots,k$, are the optimal hedging strategies for the pure mortality risks, and $\xi^{\rm cr}_j$, $j=1,\ldots,l$, are the optimal hedging strategies for the correlation risks. Even though our results can be extended to more general quadratic hedging approaches, we opted to focus on the F\"ollmer-Sondermann's method to well illustrate our main ideas. The literature addressing this objective becomes quite rich in the last decade, while the existing literature assumes assumptions on the triplet $(\mathbb F, S,\tau)$ that can be translated, in one way or another, to a sort of independence and/or no correlation between the financial market -represented by the pair $(\mathbb F,S)$- and the mortality represented by the death time $\tau$. This independence feature, with its various degree, has been criticized in the literature by both empirical and theoretical studies. In fact, a recent stream of financial literature highlights several links between demography and financial variables when dealing with longevity risk, see \cite{barrieuetal12,bisetti12} and references therein. \\

We have two main contributions  that are intimately related to each other and that realize the aforementioned main objective by giving a rigorous and precise formulation for (\ref{StrategyDecomposition}).  Our first main contribution lies in quantifying, as explicit as possible, the effect of mortality uncertainty on the risk-minimizing strategy by determining the optimal strategy in the enlarged  filtration $\mathbb G$ in terms of strategies in the smaller filtration $\mathbb F$. Our second main contribution resides in finding risk-minimizing strategies with securization by investing in stocks and one (or more)  insurance contracts such as longevity bonds.

Concerning the literature about the risk-minimization with or without mortality securitization, we cite \cite{barbarin08, barbarin09, biaginibotero15,  biaginirheinlander16, biaginischreiber13, dhalmelchiormoller08, dhalmoller06, moller98, moller01} and the references therein to cite few. In \cite{dhalmoller06,moller98, moller01}, the authors assume independence between the financial market and the insurance model, a fact that was criticized in  \cite{dhaenkukush13}. The works of \cite{biaginibotero15, biaginirheinlander13, biaginischreiber13} assume `the  H-hypothesis', which guarantees that the mortality has no effect on the martingale structure at all (i.e.,  every $\mathbb F$-martingale remains a $\mathbb G$-martingale). This assumption can be viewed as a  strong no correlation condition between the financial market and the mortality. In \cite{barbarin09}, the author weakens this assumption by considering the following two assumptions:
\begin{equation}\label{assumption1}
\mbox{Either}\ \tau\ \mbox{avoids}\ \mathbb F\mbox{-stopping times or all}\ \mathbb F\mbox{-martingales are continuous,}
\end{equation}
and 
\begin{equation}\label{assumption2} 
M^{\mathbb G}\  \mbox{is given by}\ M^{\mathbb G}_t:=\mathbb{E}[h_{\tau}\  \big|\ {\cal G}_t]\  \mbox{where}\ h\ \mbox{is}\ \mathbb F\mbox{-predictable with suitable integrability.}
\end{equation}
However,  these assumptions are also  very restrictive. It is clear that for the popular  and simple discrete time market models the assumption (\ref{assumption1}) fails. Furthermore for most models in insurance (if not all), a Poisson process is an important component in the modelling,  and hence for these models the second part of assumption (\ref{assumption1}) fails, while its first part can be viewed as a kind of `independence' assumption between the mortality (i.e., the random time $\tau$) and the financial market (i.e., the pair $(\mathbb F, S)$). In \cite{gerbershiu13} and the references therein, the authors treat many death-related claims and liabilities in (life) insurance whose payoff process $h$ fails (\ref{assumption2}).
 In \cite{barbarin09, biaginibotero15, biaginischreiber13, dhalmoller06, moller98}, the author assumed that the mortality has a hazard rate
process, mimicking the intensity-based approach of credit risk, while in \cite{barbarin08} the author uses the interest rate modelling of Heath-Jarrow-Morton. Up to our knowledge, all the literature considers the Brownian setting for the financial market except \cite{barbarin09}.\\

\noindent This paper contains four sections, including the current section, and an appendix. The aim of the next section (Section \ref{section2})  lies in introducing the mathematical model, the optional martingale representation, and the F\"ollmer-Sondermann optimization criterion. The  third and the fourth sections are the principal innovative sections of the paper and deal with quadratic hedging for mortality/longevity risks, in the spirit of F\"ollmer-Sondermann, in the cases where mortality/longevity securitization is incorporated or not.  For the sake of easy exposition, the proof of an intermediate technical lemma is relegated to the  appendix.

%%%%%%%%%%%%%%%%%%%%%%%%%%%%%%%%%%%%%%%%%%%%%%%%%%%%%%%%%%%%%%%%%%%%%%%%%%%
%%%%%%%%%%
\section{Mathematical model and preliminaries}\label{section2}
This section presents our mathematical model, which is constituted by an initial market model and a death time, and recalls our optional martingale representation result that we use throughout the paper. We conclude this section by briefly reviewing the quadratic hedging criterion of F\"ollmer-Sondermann. 
\subsection{Time of death, enlargement of filtration and a martingale representation theorem}
In addition to the initial market model $\left(\Omega, {\cal F},S,\mathbb F=({\cal F}_t)_{t\geq 0}, P\right)$,  we consider from now on an \({\cal F}\)-measurable random time $\tau$, that represents the \emph{time of death} of an insured, which might not be an $\mathbb F$-stopping time. The knowledge about this time of death is limited.  The right-continuous and non-decreasing process indicating whether death has occurred or not is denoted by $D$, while the enlarged filtration of $\mathbb F$ associated with the couple $(\mathbb F, \tau)$ is denoted by $\mathbb G$, and they are defined by
\begin{equation}\label{ProcessDandG}
D := I_{\Lbrack\tau,+\infty\Lbrack},\ \ \ \mathbb G:=({\cal G}_t)_{t\geq 0},\ \ \ {\cal G}_t=
\cap_{s>0}\left({\cal F}_{s+t}\vee\sigma\left(D_{u},\ u\leq s+t\right)\right).
\end{equation}
Thus, starting from the filtration \(\mathbb F\), which represents the flow of public information, $\mathbb G$ is the \emph{progressively enlarged} filtration by incorporating the information included in the process \(D\). \(\mathbb G\) is the smallest filtration which contains \(\mathbb F\) and makes \(\tau\) a \(\mathbb G\)-stopping time. 

We recall some notations that we will use throughout the paper.  For any filtration  $\mathbb H\in \{\mathbb F,\mathbb G\}$, we denote by ${\cal A}(\mathbb H)$ (respectively ${\cal M}(\mathbb H)$) the set
of $\mathbb H$-adapted processes with $\mathbb H$-integrable variation (respectively that are $\mathbb H$-uniformly integrable martingales).
For any process $X$,  $^{o,\mathbb H}X$  (respectively $^{p,\mathbb H}X$)  is the
$\mathbb H$-optional (respectively $\mathbb H$-predictable) projection of $X$. For an increasing process $V$,  the process $V^{o,\mathbb H}$ (respectively $V^{p,\mathbb H}$) represents its dual $\mathbb H$-optional (respectively $\mathbb H$-predictable) projection. For a filtration $\mathbb H$, ${\cal O}(\mathbb H)$, ${\cal P}(\mathbb H)$ and  $\mbox{Prog}(\mathbb H)$ denote the $\mathbb H$-optional, the $\mathbb H$-predictable and the $\mathbb H$-progressive $\sigma$-fields  respectively on $\Omega\times[0,+\infty[$. For an $\mathbb H$-semimartingale $X$, we denote by $L(X,\mathbb H)$ the set of all $X$-integrable processes in Ito's sense, and for $H\in L(X,\mathbb H)$, the resulting integral is a  one dimensional $\mathbb H$-semimartingale denoted by $H\is X:=\int_0^{\centerdot} H_u dX_u$. If ${\cal C}(\mathbb H)$ 
is a set of processes that are adapted to $\mathbb H$,
then ${\cal C}_{\loc}(\mathbb H)$ --except when it is stated otherwise-- is the set of processes, $X$,
for which there exists a sequence of $\mathbb H$-stopping times,
$(T_n)_{n\geq 1}$, that increases to infinity and $X^{T_n}$ belongs to ${\cal C}(\mathbb H)$, for each $n\geq 1$. We recall the definition of orthogonality between local martingales.
\begin{definition}
	Let $M$ and $N$ be two $\mathbb H$-local martingales. Then $M$ is said to be orthogonal to $N$ whenever $MN$ is also an $\mathbb H$-local martingale, or equivalently $[M,N]$ is an $\mathbb H$-local martingale.
\end{definition}

The public who has access to the filtration $\mathbb F$, can only get information about $\tau$ through the survival probabilities denoted by $G_t$ and $ \widetilde{G}_t$, and are given by
\begin{equation}\label{GGtildem}
G_t := {}^{o,\mathbb F}\!(I_{\Lbrack0,\tau\Lbrack})_t=P(\tau > t | {\cal F}_t),\quad  \widetilde{G}_t := {}^{o,\mathbb F}\!(I_{\Lbrack0,\tau\Rbrack})=P(\tau \ge t | {\cal F}_t),
\quad \mbox{ and } \quad \ m := G + D^{o,\mathbb F}.
\end{equation}
The processes $G$ and $\widetilde G$ are known as Az\'ema supermartingales ($G$ is right-continuous with left limits, while in general  $\widetilde G$ has right and left limits only), while $m$ is a BMO $\mathbb F$-martingale. For more details about these, we refer the reader to  \cite[paragraph 74, Chapitre XX]{dellacheriemeyer92}. 
%%\is%%\subsection{Optional martingale representation}
To derive the risk-minimizing strategies for a mortality claim we will make use of the  optional martingale representation for a $\mathbb G$-martingale  introduced in \cite{choullidaveloose2018} that states that the risk can be decomposed into three  types of risks.
%%%%%%%%%%%%%%%%%%%%%%%%%%%%%%%%%%%%%%%%%%%%%%%%%%%%%%%%%%%%
%%%%%%%%%%%%%%%%%%%%%%%%%%%%%%%%%%%%%%%%%%%%%%%%%%%%%%%%%%%%%%
 \begin{theorem}\cite[Theorem2.19]{choullidaveloose2018} \label{TheoRepresentation} Let $h\in L^2({\cal O}(\mathbb F),P\otimes D)$, and $M^h$ be given by
 	\begin{equation}\label{processesMhandJ}
 M^h_t:={}^{o,\mathbb F}\left(\int_0^{\infty}h_udD_u^{o,\mathbb F}\right)_t={\mathbb E}\Bigl[
 \int_0^{\infty}h_udD_u^{o,\mathbb F}\mid \mathcal{F}_t \Bigr].
 	\end{equation}
 	Then the $\mathbb G$-martingale $H_t := \!\! \ ^{o,\mathbb G}(h_\tau)_t=\E[h_\tau | {\cal G}_t]$ admits the following representation.
 	\begin{equation}\label{Gdecomposition}
 	H - H_0 ={\frac{I_{\Rbrack 0,\tau\Rbrack}}{G_-}}\is  \widehat{M^h}
 	-{\frac{M^h_{-} - (h\is D^{o,\mathbb F})_{-}}{G_{-}^2}} I_{\Rbrack 0,\tau\Rbrack}\is \widehat{m} +
 	{\frac{hG-M^h +h\is D^{o,\mathbb F}}{G}}I_{\Rbrack 0,R\Lbrack}\is N^{\mathbb G},
 	\end{equation}
 	where $	R:= \inf \{t\ge 0: G_t =0 \}$  and both processes 
 	\begin{eqnarray*}
	\left(hG-M^h +h\is D^{o,\mathbb F}\right){{I_{\Rbrack 0,R\Lbrack}}\over{G_{-}}}\is N^{\mathbb G}\ \mbox{and}\quad
 	{{I_{\Rbrack 0,\tau\Rbrack}}\over{G_{-}}}\is  \widehat{M^h}
 	-\left(M^h_{-} - (h\is D^{o,\mathbb F})_{-}\right){{I_{\Rbrack 0,\tau\Rbrack}}\over{G_{-}^{2}}}\is \widehat{m}\end{eqnarray*}
 	are square integrable  and orthogonal martingales under $\mathbb G$. 
 \end{theorem}
 The first term in the RHS of \eqref{Gdecomposition} represents the `pure' financial risk, while the second term represents the risk resulting from correlation between the market model and the death time $\tau$. Both the first and the second terms are expressed in terms of  $\mathbb G$-local martingales derived in \cite{aksamitetal15}. Below, we recall this class of local martingales.
\begin{theorem} \cite[Theorem 3]{aksamitetal15} For any $\mathbb F$-local martingale $M$, the following
	\begin{equation} \label{processMhat}
	\widehat{M} := M^\tau -{\widetilde{G}}^{-1} I_{\Rbrack 0,\tau\Rbrack}\is [M,m] +
	I_{\Rbrack 0,\tau\Rbrack}\is\Big(\Delta M_{\widetilde R} I_{\Lbrack\widetilde R,+\infty\Lbrack}\Big)^{p,\mathbb F},
	\end{equation}
	is a $\mathbb G$-local martingale. Here
	\begin{equation} \label{timetildeR}
	R:= \inf \{t\ge 0: G_t =0 \},\ \ \ \  \mbox{and}\ \ \ \ \ \widetilde R := R_{\{\widetilde{G}_R=0<G_{R-}\}} =
	R I_{\{\widetilde{G}_R=0<G_{R-}\}} + \infty I_{\{\widetilde{G}_R=0<G_{R-}\}^c}.
	\end{equation}
\end{theorem}
The third term in the RHS of \eqref{Gdecomposition} models the {\it pure mortality risk of type one} (see \cite[Theorem 2.13]{choullidaveloose2018} for details) where the process $N^{\mathbb G}$ is given by  
	\begin{equation} \label{processNG}
	N^{\mathbb G}:=D - \widetilde{G}^{-1} I_{\Rbrack 0,\tau\Rbrack}\is D^{o,\mathbb  F},
	\end{equation}
 which is a $\mathbb G$-martingale with integrable variation. 
 This pure mortality risk is called a pure default martingale in \cite[Definition 2.2]{choullidaveloose2018}, and it quantifies the uncertainty in $\tau$ (or equivalently in  $D$ defined in (\ref{ProcessDandG})) that cannot be seen through $\mathbb F$. For other types of pure mortality risks (local martingale) and for further details about pure mortality (or default) local martingales, we refer the reader to \cite{choullidaveloose2018}. Our decomposition \eqref{Gdecomposition} extends \cite{blanchetjeanblanc04} to an arbitrary general pair $(\mathbb F, \tau)$ and to the case where $h$ is $\mathbb F$-optional, as is the case for some examples in \cite{gerbershiu13}.

%%%%%%%%%%%%%%%%%%%%%%%%%%%%%%%%%%%%%%%%%%%%%%%%%%%%%
%%%%\subsection{Risk's decomposition for some mortality/longevity securities}\label{subsection4Riskdecomp}

%\begin{corollary}\label{independCase}

 %%%%%%%%%%%%%%%%%%%%%%%%%%%%%%%%%%%%%%%%%%%%%%%%%%%%%%%%%%%%%%%%%%%%%%%
 %%%%%%%%%%%%%%%%%%%%%%%%%%%%%%%%%%%%%%%%%%%%%%%%%%%%%%%%%%%%%%%%%%%%%%
 
\subsection{The quadratic risk-minimizing method}\label{subsection4Riskmini}
In this subsection, we quickly review the main ideas of risk-minimizing strategies, a concept that
 was introduced in \cite{follmersondermann86}  for financial contingent claims
 and extended in \cite{moller01} for insurance payment processes. Note that \cite{follmersondermann86} assumed that  the discounted
 risky asset is a square-integrable martingale under
    the original measure $P$. In \cite{schweizer01}, the results are proved under the  weaker assumption
     that $X$ is only a local $P$-martingale, that does not need to be locally square integrable. Throughout this subsection, 
     we consider given an $\mathbb H$-adapted process $X$ with values in  $\mathbb{R}^d$ representing the discounted assets' price process.  $\mathbb H$ is any filtration satisfying the usual condition (usually $\mathbb H\in\{\mathbb F, \mathbb G\}$). Throughout the paper, we denote by $x^{tr}y$ the inner product of $x$ and $y$, for any $x,y\in{\mathbb R}^d$. 
\begin{definition}\label{riskmini} Suppose that $X\in {\cal M}_{loc}(\mathbb H)$.

{\rm{(a)}} An $0$-admissible  trading strategy is any pair $\rho:=(\xi,\eta)$ where $\xi\in L^2(X^T)$ with $L^2(X^T)$
 the space of all $\mathbb{R}^d$-valued predictable processes $\xi$ such that
\[
\|\xi \|_{L^2(X)}:=\left(\mathbb{E}\left[\int_0^T\xi^{tr}_ud[X]_u\xi_u\right]\right)^{1/2} <\infty ,
\]
and $\eta$  is a real-valued adapted process such that the discounted value process
 \begin{equation}\label{Valurprocess}
  V(\rho)=\xi^{tr} X^T+\eta\ \ \mbox{is  right-continuous and square-integrable, and }  V_T(\rho)=0,\ \ \ \  P\mbox{-a.s.}.
 \end{equation}
 {\rm{(b)}} $\rho$ is called risk-minimizing  for the square integrable
  $\mathbb{H}$-adapted payment process $A=(A_t)_{t\geq 0}$, if it is  an $0$-admissible strategy and for any $0$-admissible strategy $\tilde{\rho}$, we have
  \begin{equation}\label{risminimistrategy}
  R_{t\wedge T}(\rho) \le R_{t\wedge T}(\tilde{\rho})\quad P\mbox{-a.s.\ for every } t\geq 0,
\end{equation}
where
\[
 R_t(\rho) := \E[(C_T(\rho) - C_{t\wedge T}(\rho))^2 \mid {\cal F}_t]\ \ \mbox{and}\ \
C(\rho) := V(\rho) - \xi \is X^T + A^T.
\]
\end{definition}

It is known in the literature that the Galtchouk-Kunita-Watanabe decomposition  (called hereafter GKW decomposition) plays a central role in determining the risk-minimizing strategy. 

\begin{theorem}\label{GKWdecomposition} Let $M, N\in {\cal M}_{loc}^2(\mathbb H)$. Then there exist
 $\theta\in L^2_{loc}(N)$ and $L\in {\cal M}_{0,loc}^2(\mathbb H)$ such that
\begin{equation}\label{GKW}
M=M_0+\theta\is N+L,\ \ \ \ \mbox{and}\ \ \ \langle {N},L\rangle^{\mathbb H}\equiv 0.
\end{equation}
Furthermore, $M\in {\cal M}^2(\mathbb H)$ if and only if $M_0\in L^2({\cal F}_0,P)$, $\theta\is N\in {\cal M}_0^2(\mathbb H)$ and $L\in {\cal M}_0^2(\mathbb H)$.
\end{theorem}

For more about GKW decomposition, we refer the reader
 to \cite{anselstricker92, choullistricker96},
  and the references therein. The following theorem was proved for a single payoff in \cite{schweizer01}, and was extended to payment processes in \cite{moller01}.

\begin{theorem}\label{RiskMini}
Suppose that $X\in {\cal M}_{loc}(\mathbb H)$, and let $A=(A_t)_{t\geq 0}$ be the payment process that is square integrable. Then the following holds.\\
{\rm{(a)}} There exists a unique risk-minimizing strategy $\rho^*=(\xi^*,\eta^*)$  for $A$ given by
\begin{equation}\label{strategies}
\xi^*:=\xi^A\ \ \ \ \mbox{and}\ \ \ \eta^*_t:=\E[A_T-A_{t\wedge T}\mid  {\cal H}_t]-{\xi^*_t}^{tr} X_{t\wedge T},
\end{equation}
where $(\xi^A, L^A)=(\xi^A I_{\Lbrack0,T\Rbrack}, (L^A)^T)$ is the pair resulting from the GKW decomposition of $\E[A_T\mid {\cal F}_t]$ with respect
 to $X$ with $\xi^A\in L^2(X^T)$ and $L^A\in \mathcal{M}^2_0(\mathbb H)$ satisfying  $\langle L^A, \theta \is X\rangle\equiv 0$, for all $ \theta\in L^2(X)$.\\
{\rm{(b)}} The remaining (undiversified) risk is  $L^A$, while the optimal cost, risk and value processes are 
 \begin{equation}\label{Cost}
 C_t(\rho^*) = \E[A_T\mid {\cal H}_0]  + L^{A}_t ,\quad
R_t(\rho^*) = \E[(L^{A}_T- L^{A}_t)^2\mid {\cal H}_t], \quad \mbox{and}\quad V_t(\rho^*)=\E[A_T-A_{t\wedge T}\mid  {\cal H}_t].
\end{equation}
\end{theorem}
%%%%%%%%%%%%%%%%%%%%%%%%%%%%%%%%%%%%%%%%%%%%%%%%%%%%%%%%%%%%%%%%%%%%%%%
%%%%%%%%%%%%%%%%%%%%%%%%%%%%%%%%%%%%%%%%%%%%%%%%%%%%%%%%%%%%%%%%%%%%%%%%%

 \noindent The next two sections contain the main two contributions of this paper and deal with hedging mortality liabilities {\it \`a la} F\"ollmer-Sondermann.
%%%%%%%%%%%%%%%%%%%%%%%%%%%%%%%%%%%%%%%%%%%%%%
%%%%%%%%%%%%%%%%%%%%%%%%%%%%%%%%%%%%%%%%%%%%%%%%%%%%%%%%%%%%%%%%%%

\section{Hedging mortality risk without securitization} \label{secRiskMinimisationFG}

In this section, we hedge the mortality liabilities without mortality securitization. In this context, our aim
lies in quantifying -as explicit as possible- the effect of mortality uncertainty on the risk-minimizing strategy.
 This will be achieved by determining the $\mathbb G$-optimal strategy in terms of $\mathbb F$-strategies for a large class of mortality contracts. This section contains three subsections. The first subsection deals with the general setting, while the second and third subsections illustrates further the obtained results in the first subsection on particular cases of mortality liabilities. {\bf Throughout the rest of the paper, we consider a  given finite time horizon $T>0$}.

%%%%%%%%%%%%%%%%%%%%%%%%%%%%%%%%%%%%%%%%%%%%%%
%%%%%%%%%% SUBSECTION 4.1
%%%%%%%%%%%%%%%%%%%%%%%%%%%%%%%%%%%%%%%%%%%%%%%
%%%%%%%%%%%%%%%%%%%%%%%%%%%%%%%%%%%%%%%%%%%%%
%%%%%%%% SUBSECTION 4.2
%%%%%%%%%%%%%%%%%%%%%%%%%%%%%%%%%%%%%%%%%%%%%%%%

\subsection{\(\mathbb G\)-Optimal strategy in terms of \(\mathbb F\)-optimal strategies: The general formula}\label{Subsection4results}

This subsection considers a portfolio consisting of life insurance
 liabilities depending on the random time of death \(\tau\) of a single insured.
  For the sake of simplicity, we assume that the policyholder of a contract is the insured itself.
   In the financial market, there is a risk-free asset and a multidimensional risky asset at hand.
    The price of the risk-free asset follows a strictly positive, continuous process of finite variation, and the risky asset follows
    a real-valued RCLL  \(\mathbb F\)-adapted stochastic process. The discounted value of the risky asset is denoted by \(S \). 
    In order to reach our goal of expressing the \(\mathbb G\)-optimal strategy in terms of \(\mathbb F\)-strategies via the F\"ollmer-Sondermann method,  we need to assume on the pair $(S,\tau)$ the following conditions.
\begin{equation}\label{mainassumpOn(X,tau)}
 S\in{\cal M}_{\loc}^2(\mathbb F), \quad \langle S, m\rangle^{\mathbb F}\equiv 0,
 \quad  \mbox{and}\quad \{\Delta S\not=0\}\cap\{\widetilde G=0<G_{-}\}=\emptyset. \end{equation}

The assumption $\{\Delta S\not=0\}\cap\{\widetilde G=0<G_{-}\}=\emptyset$ guarantees the structure
   conditions for $\left(S^{\tau},\mathbb G\right)$, and hence the quadratic risk-minimization problem can have a solution for this model. This assumption holds when the hazard rate (i.e., $G>0$) exists for instance. The conditions $S\in{\cal M}_{\loc}^2(\mathbb F)$ and $\langle S, m\rangle^{\mathbb F}\equiv 0$  are dictated by the method used for risk minimization. In particular, the assumption $\langle S, m\rangle^{\mathbb F}\equiv 0$ implies that the risk $m$ cannot be hedged in the model $(S, \mathbb F)$. The risk-minimizing method is the quadratic
  hedging approach \`a la F\"ollmer and Sondermann, which requires that the discounted price processes for the underlying assets are locally
   square integrable martingales. In fact, under these two latter conditions,  both models $(S, P,\mathbb F)$ and $(S^{\tau},P,\mathbb G)$ are local martingales, and hence  the F\"ollmer-Sondermann
 method will be applied simultaneously for both models. These two assumptions in (\ref{mainassumpOn(X,tau)}) can be relaxed at the expenses of considering the quadratic hedging method considered in \cite{choullistricker98, monatstricker95}, and the references therein.
  For the risk-minimization framework of these papers, the assumption $\sup_{0\leq t\leq \cdot}\vert S_t\vert^2\in {\cal A}^+_{loc}(\mathbb F)$
  will suffice together with some ``no-arbitrage or viability" assumption on $(S,\tau)$, developed in \cite{choullideng14}.\\
  
  \noindent Our main results of this section are based essentially on the following.

\begin{lemma}\label{consequences4MainAssum}
Suppose that (\ref{mainassumpOn(X,tau)}) holds. Then the following assertions hold.\\
{\rm{(a)}} We have $S^{\tau}\in {\cal M}^2_{\loc}(\mathbb G)$.\\
{\rm{(b)}} The \(\mathbb  G\)-martingale $\widehat L$ is orthogonal to $S^{\tau}$, for any $L\in {\cal M}_{\loc}(\mathbb F)$ that is orthogonal to $S$.\\
 {\rm{(c)}} The process
 \begin{equation}\label{processU}
 U:=I_{\{G_{-}>0\}}\is [S,m]
 \end{equation}
 is an $\mathbb F$-locally square integrable local martingale. Thus, there exist $\varphi^{(m)}\in L^2_{\loc}(S,\mathbb F)$
  and $L^{(m)}\in {\cal M}^2_{0,\loc}(\mathbb F)$ orthogonal to $S$ such that
\begin{equation}\label{decompositionUandhatU}
U=\varphi^{(m)}\is S +L^{(m)},\ \ \ \ \mbox{and}\ \ \ \ \
 \Rbrack 0, \tau \Rbrack \subseteq \{G_- > 0\} \subseteq \{G_- + \varphi^{(m)} > 0 \}, \quad P\text{-a.s.}.
\end{equation}
 {\rm{(d)}} We have $ \widehat U=G_{-}{\widetilde G}^{-1}I_{\Rbrack 0,\tau\Rbrack}\is U$ and 
\begin{equation}\label{decompositionXhat}
(G_{-}+\varphi^{(m)})\is\widehat {S}=G_{-}\is S^{\tau}-\widehat{L^{(m)}}.
\end{equation}
\end{lemma}
The proof of this lemma is postponed to Appendix \ref{AppendixproofLemma} for the sake of simple exposition. Below,
we state our main results  of this subsection.

\begin{theorem}\label{riskMininStrategy}
Suppose that (\ref{mainassumpOn(X,tau)}) holds, and let $h\in L^2\left({\cal O}(\mathbb F),P\otimes D\right)$.
Then the following hold.\\
{\rm{(a)}} The risk-minimizing strategy for the mortality claim $h_{\tau}$, at term \(T\) under the model $(S^{\tau},\mathbb G)$,
 is denoted by $\xi^{(h,\mathbb G)}$ and is given by
\begin{equation}\label{riskminG}
\xi^{(h,\mathbb G)}:=\xi^{(h,\mathbb F)}\big(G_{-}+\varphi^{(m)}\big)^{-1}I_{\Rbrack 0,\tau\Rbrack}.
\end{equation}
Here $\xi^{(h,\mathbb F)}$ is the risk-minimizing strategy under $(S,\mathbb F)$ for the claim
 $\E\Big[\int_0^\infty h_u  dD_u^{o,\mathbb F} \Big| {\cal F}_T \Big]$.\\
{\rm{(b)}} The remaining (undiversified) risk for the mortality claim $h_{\tau}$, at term \(T\) under
 the model $(S^{\tau},\mathbb G)$, is denoted by $L^{(h,\mathbb G)}$ and is given by
\begin{align}\label{riskminGremaining}
L^{(h,\mathbb G)}:=& {{- \xi^{(h,\mathbb  F)} G_{-}^{-1}}\over{G_{-}+\varphi^{(m)}}}I_{\Rbrack 0,\tau\Rbrack}\is\widehat {L^{(m)}}
+{{I_{\Lbrack 0,\tau\Rbrack}}\over{G_{-}}} \is\widehat{L^{(h,\mathbb F)}}
 - {{M^h_{-}-(h\is D^{o,\mathbb F})_{-}}\over{ G_{-}^2}} I_{\Rbrack 0,\tau\wedge T\Rbrack} \is\widehat{m}\nonumber\\
 & + {\frac{Gh-M^h+h\is D^{o,\mathbb F}}{G}} I_{\Rbrack 0,R\Lbrack}\is \left(N^{\mathbb G}\right)^T.
 \end{align}
Here $L^{(h,\mathbb F)}$ is the remaining (undiversified) risk under $(S,\mathbb F)$ for the claim
$\E\left[\int_0^{\infty}h_u dD^{o,\mathbb F}_u \Big| {\cal F}_T\right]$, while  \({M^h}\) and  (\(\varphi^{(m)}\), 
\(L^{(m)}\)) follow from {\eqref{processesMhandJ}} and (\ref{decompositionUandhatU}) respectively.\\
 {\rm{(c)}} The value of the risk-minimizing portfolio \(V(\rho^{*,\mathbb{G}})\) under $(S^{\tau},\mathbb G)$ is given by
\begin{equation} \label{ValueRiskMiniGeneral}
V(\rho^{*,\mathbb{G}}) =  h_\tau I_{\Lbrack\tau,+\infty\Lbrack}  +
	G^{-1}{\; }^{o,\mathbb F}\!\left(h_\tau  I_{\Rbrack 0,\tau\Lbrack}\right) I_{\Rbrack 0,\tau\Lbrack} -h_\tau I_{\Lbrack T\Rbrack}.
\end{equation}
\end{theorem}

The life insurance liabilities where the claim \(h_\tau\) is determined by an optional process \(h\) appear, typically,
in the form of unit-linked insurance products.  In these type of term insurance contracts, the insurer pays an amount \(K_\tau\)
at the time of death \(\tau\), if the policyholder dies before or at the term of the contract \(T\), or equivalently the discounted
payoff is \(I_{\{\tau\le T\}} {K}_\tau\), where ${K}\in L^2({\cal O}(\mathbb F), P\otimes D)$. As a result, the payoff process
for this case is
\begin{equation}\label{h2}
h_t:= I_{\{t\le T\}} {K}_t,\ \ \ \ \  \mbox{where}\ \ \ K\in {\cal O}(\mathbb F),\
 \E\left[\vert K_{\tau}\vert ^2 I_{\{\tau<+\infty\}}\right]<+\infty.
\end{equation}
For this case, the pair $\left(\xi^{(h,\mathbb F)}, L^{(h,\mathbb F)}\right)$ in (\ref{riskminG})-(\ref{riskminGremaining})
 are the minimizing strategy and the remaining risk for the payoff $\int_0^T K_t dD^{o,\mathbb F}_t$ under the model
  $\left(S,\mathbb F\right)$, while the value process $V(\rho^{*\mathbb G})$ under the model $(S^{\tau},\mathbb G)$ becomes
\begin{equation}\label{ValueTI}
V(\rho^{*\mathbb G})=		G^{-1}\ ^{o,\mathbb F}\!\left(h_\tau  I_{\Lbrack 0,\tau\Lbrack}\right) I_{\Lbrack 0,\tau\Lbrack}.
\end{equation}
Hereto, by considering the payment process $A=K_{\tau}I_{\Lbrack \tau, +\infty\Lbrack}$, we derive $A_T=h_{\tau}$ and  for $t\in [0,T]$ 	
$$A_T-A_t=I_{\{\tau\leq T\}}K_{\tau}-I_{\{\tau \leq t\}}K_{\tau}=I_{\{t<\tau\}}I_{\{\tau\leq T\}}K_{\tau}=I_{\{t<\tau\}} h_{\tau}.$$
 Thus $V(\rho^{*\mathbb G})={\;}^{o,\mathbb G}(h_{\tau}I_{\Lbrack 0,\tau\Lbrack}) $ which is exactly the second term on the RHS of
  \eqref{ValueRiskMiniGeneral}. This extends the results of  \cite{biaginischreiber13}, where the authors assume that $K$ does not jump at $\tau$ (i.e., so that
$I_{\{\tau\leq t\}}K_{\tau}=I_{\{\tau\leq t\}}K_{\tau -} $ ), and hence they can treat it as a predictable case. More precisely they consider a life insurance payment process 
$A$ with $A_t=I_{\{\tau\leq t\}}\bar{A}_t$ with $\bar{A}$ a predictable process given by $\bar{A}_t=K_{t-}$ for $t\in ]0,T]$. 

\begin{proof}[Proof of Theorem \ref{riskMininStrategy}]:  By applying Theorem \ref{TheoRepresentation} to $H$, where  $H_t=\E[h_{\tau}\mid {\cal G}_t]$ is a $\mathbb G$-square integrable martingale, we get the decomposition \eqref{Gdecomposition}.

Thus, the main idea of the proof lies in applying the risk-minimization for the risk $M^h=\ ^{o,\mathbb F}\!\left(\int_0^{\infty}h_udD^{o,\mathbb F}_u\right)$ under the model $(S,\mathbb F)$,
 and using Lemma \ref{consequences4MainAssum} to get the explicit form of the $\mathbb G$-strategy. Notice that the risk $m$ cannot be hedged under the model $(S,\mathbb F)$ due to the second assumption in \eqref{mainassumpOn(X,tau)}. Once the strategy is described, we will prove that this strategy indeed belongs to $L^2(S^{\tau},\mathbb G )$ (i.e., it is `admissible') afterwards.
    This will follow from proving $M^h$ is a square integrable $\mathbb F$-martingale.
     This is the aim of the first step below, while the second step describes the  $\mathbb G$-strategy explicitly and locally on a sequence of subsets that increases to $\Omega\times [0,+\infty)$. The third (last) step proves the admissibility of the $\mathbb G$-strategy and ends the proof of the theorem.\\
{\bf Step 1)} Let $K\in L^{\infty}({\cal F}_{\infty},P)$, and put the $\mathbb F$-martingale $K_t:=\E[K \mid {\cal F}_t]$. Then, we derive
\begin{align*}
\E\left(K\int_0^{\infty} h_u dD^{o,\mathbb F}_u\right)&= \E\left(\int_0^{\infty} K_uh_u dD^{o,\mathbb F}_u\right)\leq
\E\left(\int_0^{\infty} \sup_{0\leq t\leq u}\vert K_t\vert \vert h_u\vert dD^{o,\mathbb F}_u\right)\\
&=  \E\left(\int_0^{\infty} \sup_{0\leq t\leq u}\vert K_t\vert \vert h_u\vert dD_u\right)=
\E\left(\sup_{0\leq t\leq \tau}\vert K_t\vert \vert h_{\tau}\vert I_{\{\tau<+\infty\}}\right)\\
&\leq  \sqrt{\E(h_{\tau}^2I_{\{\tau<+\infty\}})}\sqrt{\E\left(\sup_{t\geq 0}\vert K_t\vert^2\right)}\leq 2
\sqrt{\E(h_{\tau}^2I_{\{\tau<+\infty\}})}\sqrt{\E\left(\vert K\vert^2\right)},
\end{align*}
where the last inequality follows from Doob's inequality. Thus, this proves that
$\int_0^{\infty}h_u dD^{o,\mathbb F}_u$ is a square integrable random  variable for any $h\in L^2({\cal O}({\mathbb F}), P\otimes D)$. As a result, $M^h\in{\cal M}^2(\mathbb F)$. \\
{\bf Step 2)} By applying Theorem \ref{GKWdecomposition}  to  the pair $(M^h, S)$ of elements of $ {\cal M}_{\loc}^2(\mathbb F)$, we deduce the existence of the pair  $(\xi^{(h,\mathbb F)} ,L^{(h,\mathbb F)})$ such that 
\begin{equation}\label{mainequation100}
M^h = M_0^h + \xi^{(h,\mathbb F)} \is S + L^{(h,\mathbb F)}.
\end{equation}
Hence $\xi^{(h,\mathbb F)}$ is the risk-minimizing strategy and $L^{(h,\mathbb F)}$ is the remaining risk, under $(S,\mathbb F)$ for the 
claim $\E[\int_0^{\infty}h_u dD^{o,\mathbb F}_u\mid \mathcal{F}_T]$ at term \(T\). As a result, we get $\widehat{M^h} = \xi^{(h,\mathbb F)} \is{\widehat S} + \widehat{L^{(h,\mathbb F)}}$, and by inserting this into (\ref{Gdecomposition}) we obtain 
\begin{equation}\label{Gdecomposition2}
H = H_0 + {\frac{\xi^{(h,\mathbb F)}}{G_{-}}} I_{\Rbrack 0,\tau\Rbrack} \is\widehat{S}
 +{\frac{I_{\Rbrack 0,\tau\Rbrack}}{ G_{-}}} \is\widehat{L^{(h,\mathbb F)}}
   - {\frac{M^h_{-}-(h\is D^{o,\mathbb F})_{-}}{ G_{-}^2}}I_{\Rbrack 0,\tau\Rbrack}\is\widehat{m}
    +{\frac{hG-M^h+h\is D^{o,\mathbb F}}{ G}} I_{\Lbrack 0,R\Lbrack}\is \left( N^{\mathbb G}\right)^T.
\end{equation}
Put
\begin{equation}\label{Gamman}
\Sigma_n:=\left( \{\vert \xi^{(h,\mathbb F)}\vert\leq n\ \&\ G_{-}+
\varphi^{(m)}\geq 1/n\}\cap\Rbrack 0,\tau\Rbrack\right)\bigcup\ \Rbrack\tau,+\infty\Lbrack,
\end{equation}
and utilize (\ref{decompositionXhat}) to derive
\begin{align*}
I_{\Sigma_n}\is H& =  {{\xi^{(h,\mathbb F)}}\over{G_{-}+\varphi^{(m)}}} I_{\Rbrack 0,\tau\Rbrack\cap \Sigma_n} \is S^{\tau}
 - {{\xi^{(h,\mathbb F)} G_{-}^{-1}}\over{G_{-}+\varphi^{(m)}}}I_{\Rbrack 0,\tau\Rbrack\cap \Sigma_n}\is\widehat{ L^{(m)}} \\
&\quad + {{I_{\Rbrack 0,\tau\Rbrack\cap \Sigma_n}}\over{G_{-}}}  \is\widehat{L^{(h,\mathbb F)}}
 - {{M^h_{-}-(h\is D^{o,\mathbb F})_{-}}\over{ G_{-}^2}}I_{\Rbrack 0,\tau\Rbrack\cap \Sigma_n} \is\widehat{m}
  + {{hG-M^h+h\is D^{o,\mathbb F}}\over{ G}}I_{\Lbrack 0,R\Lbrack}I_{\Sigma_n} \is \left( N^{\mathbb G}\right)^T\\
&=: \xi^{(n,\mathbb G)}\is S^{\tau}+L^{(n,\mathbb G)},
\end{align*}
where
\begin{align*}
	\xi^{(n,\mathbb G)}&:= \xi^{(h,\mathbb F)}\big(G_{-}+\varphi^{(m)}\big)^{-1} I_{\Rbrack 0,\tau\Rbrack\cap \Sigma_n}\ \ \mbox{and}\\
	\\
	L^{(n,\mathbb G)}&:= {{-\xi^{(h,\mathbb F)}I_{\Rbrack 0,\tau\Rbrack\cap \Sigma_n}}\over{ G_{-}\big(G_{-}
+\varphi^{(m)}\big)}}\is\widehat{L^{(m)}} +{{I_{\Rbrack 0,\tau\Rbrack\cap \Sigma_n}}\over{ G_{-}}}
 \is  \widehat{L^{(h,\mathbb F)}}
 - {{M^h_{-}-(h\is D^{o,\mathbb F})_{-}}\over{ G_{-}^2}}I_{\Rbrack 0,\tau\Rbrack\cap \Sigma_n} \is\widehat{m}  \\
 &\quad  +  {{hG-M^h+h\is D^{o,\mathbb F}}\over{ G}} I_{\Lbrack 0,R\Lbrack}I_{\Sigma_n} \is \left( N^{\mathbb G}\right)^T.
 \end{align*}
{\bf Step 3)} Here we prove that $\xi^{(h,\mathbb G)}:=\displaystyle\lim_{n\longrightarrow+\infty} \xi^{(n,\mathbb G)}$
 belongs in fact to $L^2(S^{\tau},\mathbb G)$. To this end, we remark that $[\xi^{(n,\mathbb G)}\is S^{\tau},L^{(n,\mathbb G)}]=
 \xi^{(n,\mathbb G)}\is [S^{\tau},L^{(n,\mathbb G)}]$
is a $\mathbb G$-local martingale, and we consider a sequence of
$\mathbb G$-stopping times $(\sigma(n,k))_{k\geq 1}$ that goes to infinity with $k$ such that
$[\xi^{(n,\mathbb G)}\is S^{\tau},L^{(n,\mathbb G)}]^{\sigma(n,k)}$ is a uniformly integrable martingale. Then, we get
\[
\E\big[ [I_{\Sigma_n}\is H ]_{\sigma(n,k)}\big]= \E\big[ [\xi^{(n,\mathbb G)}\is S^{\tau} ]_{\sigma(n,k)}\big]+
\E \big[ [L^{(n,\mathbb G)} ]_{\sigma(n,k)} \big] \leq \E \big[ [H, H]_{\infty}\big] < +\infty.
\]
Thus, by combining this with Fatou's lemma (we let $k$ goes to infinity and then $n$ goes to infinity afterwards) and the fact
 $\xi^{(n,\mathbb G)}$ converges pointwise to $\xi^{(h,\mathbb G)}$, we conclude that $\xi^{(h,\mathbb G)}\in L^2(S^{\tau},\mathbb G)$, and $\xi^{(n,\mathbb G)}\is S^{\tau}$ converges to
  $\xi^{(h,\mathbb G)}\is S^{\tau}$ in ${\cal M}^2(\mathbb G)$. Since $I_{\Sigma_n}\is H$ converges to
   $H-H_0$ in the space of ${\cal M}^2(\mathbb G)$, we conclude that $L^{(n,\mathbb G)}$ converges in the space ${\cal M}^2(\mathbb G)$,
    and its limit $L^{(h,\mathbb G)}$ is orthogonal to $S^{\tau}$. As a result, we deduce
$\xi^{(h,\mathbb F)}\left(G_{-}+\varphi^{(m)}\right)^{-1}I_{\Rbrack 0,\tau\Rbrack}$ is $\widehat{L^{(m)}}$-integrable
and the resulting integral is a $\mathbb G$-local martingale.  This proves assertions (a) and (b), while assertion (c)
  is immediate from the fact that the $\mathbb G$-payment process corresponding to the claim $h_{\tau}$ at term $T$ is
   $A_t= I_{\{t=T\}}h_{\tau}$ and the value process of the portfolio is given by
\[
V_t(\rho^{*,\mathbb{G}}) = \E[A_T|{\cal G}_t]-A_t=\E[h_{\tau}|{\cal G}_t]- I_{\{t=T\}}h_{\tau} =H_t- I_{\{t=T\}}h_{\tau}.
\]
where the $\mathbb G$-martingale $H$ is decomposed as 
\begin{equation}\label{DecompHexplicit}
H_t=h_{\tau}I_{\{\tau\leq t\}}+\frac{I_{\{t<\tau\}}}{G_t}\mathbb{E}\left[h_\tau I_{\{t<\tau\}}\mid\mathcal{F}_t\right].
\end{equation}
This ends the proof of this theorem.
\end{proof}

Below, we elaborate the results of Theorem \ref{riskMininStrategy} in this setting where the payoff process $h$ is $\mathbb F$-predictable.

\begin{corollary}\label{riskMininStrategyPredictable}
Suppose that (\ref{mainassumpOn(X,tau)}) holds, and consider $h\in L^2\left({\cal P}(\mathbb F),P\otimes D\right)$. Let $m^h$ be given by 
\begin{equation} \label{processesmhandJ}
m^h := \ ^{o,\mathbb F}\Big(\int_0^\infty h_u  dF_u \Big),\ \ \mbox{where}\ \ \ F:=1-G.
\end{equation} 
 Then the risk-minimizing strategy and the remaining risk for the mortality claim $h_{\tau}$, at term \(T\) under $(S^{\tau},\mathbb G)$, 
 are denoted by $\xi^{(h,\mathbb G)}$ and $L^{(h,\mathbb G)}$ and are given by
\begin{align}
\xi^{(h,\mathbb G)}&:= \xi^{(h,\mathbb F)}\big(G_{-}+\varphi^{(m)}\big)^{-1}I_{\Rbrack 0,\tau\Rbrack},\label{riskminGPred}\\
 L^{(h,\mathbb G)} &:=   \frac{- G_{-}^{-1}\xi^{(h,\mathbb F)}}{G_{-}+\varphi^{(m)}} I_{\Rbrack 0,\tau\Rbrack}\is\widehat{ L^{(m)}}
 +{\frac{I_{\Rbrack 0,\tau\Rbrack}}{G_{-}}} \is\widehat{L^{(h,\mathbb F)}}
  	+ {\frac{hG_{-}-{m^h_{-}+(h\is F)_{-}}}{ G_{-}^2}}I_{\Rbrack 0,\tau\wedge T\Rbrack} \is\widehat{m} \nonumber \\
  	&\quad   + {\frac{hG-{m^h+h\is F}}{ G}} I_{\Rbrack 0,R\Lbrack} \is \left(N^{\mathbb G}\right)^T. \label{riskminGremainingPred}
\end{align}
Here the pair $\left(\xi^{(h,\mathbb F)}, L^{(h,\mathbb F)}\right)$ is the risk-minimizing strategy and the remaining risk,
 under $(S,\mathbb F)$ for the claim $\E\left[ \int_0^\infty h_u d F_u |  {\cal F}_T \right]$,  and $(\varphi^{(m)}, L^{(m)})$ is
 given in (\ref{decompositionUandhatU}).
\end{corollary}

\noindent The proof of this corollary mimics the  proof of Theorem \ref{riskMininStrategy}, and will be omitted.

\begin{corollary}\label{RiskMiniPseudo} Let $h\in L^2\left({\cal O}(\mathbb F),P\otimes D\right)$. Suppose that $\tau$ is a pseudo-stopping time, i.e., $\mathbb{E}(M_{\tau})=\mathbb{E}(M_0)$ for any $\mathbb F$-martingale $M$ ({\it in particular when $\tau$ is independent of ${\cal F}_{\infty}:=\sigma(\cup_{t\geq 0}{\cal F}_t)$}). Then the following hold.\\
 {\rm{(a)}} If $S\in {\cal M}_{loc}^2(\mathbb F)$, then (\ref{mainassumpOn(X,tau)}) holds.\\
 {\rm{(b)}}  Suppose that  $S\in {\cal M}_{loc}^2(\mathbb F)$. Then the risk-minimizing strategy and the remaining risk
  for the mortality claim $h_{\tau}$, at term \(T\) under $(S^{\tau},\mathbb G)$, are denoted by $(\xi^{(h,\mathbb G)},L^{(h,\mathbb G)})$ and are given by
\begin{equation*}
\xi^{(h,\mathbb G)}:={\frac{\xi^{(h,\mathbb F)}}{G_{-}}}I_{\Rbrack 0,\tau\Rbrack}\quad \mbox{and}\quad
L^{(h,\mathbb G)} :={\frac{I_{\Rbrack 0,\tau\Rbrack}}{ G_{-}}}  \is L^{(h,\mathbb F)}
+ {\frac{hG-M^h+h\is D^{o,\mathbb F}}{ G}}  I_{\Lbrack 0,R\Lbrack} \is \left(N^{\mathbb G}\right)^T.\ \ \ \ %\label{riskminGremainingPred}
\end{equation*}
Here $\left(\xi^{(h,\mathbb F)}, L^{(h,\mathbb F)}\right)$ is the pair of the risk-minimizing strategy and the remaining risk,
 under the model $(S,\mathbb F)$, for the claim
  $\E[ \int_0^\infty h_u d D^{o,\mathbb F}_u\ \big|\  {\cal F}_T ]$ at term \(T\), and \(M^h\) is  defined in {\eqref{processesMhandJ}}.
\end{corollary}

\begin{proof} Thanks to \cite{nikeghbaliyor05}, we deduce that $m\equiv  m_0$ (constant process) as soon as $\tau$ is a pseudo-stopping time. This implies that  $U$, defined in (\ref{processU}), is a null process. Therefore, we conclude that
    $$\varphi^{(m)}\equiv 0,\quad  L^{(m)}\equiv 0,\quad I_{\Rbrack 0,\tau\Rbrack}\is\widehat m\equiv 0\quad \mbox{and}\quad \widehat M=M^{\tau}\quad \mbox{ for any}\quad M\in {\cal M}(\mathbb F).$$
Then, by inserting these into (\ref{riskminG}) and (\ref{riskminGremaining}), the proof of the corollary follows immediately.
\end{proof}

 It is worth mentioning that the pseudo-stopping time model for $\tau$ covers the case when $\tau$ is independent of ${\cal F}_{\infty}$ (no correlation between the financial
market and the death time), the case when $\tau$ is an $\mathbb F$-stopping time (i.e.,  the case of full correlation between the financial market and the death time), and the case when there is arbitrary moderate correlation such as the immersion case of $\tau:=\inf\{t\geq 0\ \big|\ S_t\geq E\}$ with $E$ is a random variable that is independent of ${\cal F}_{\infty}$.
For more details about pseudo-stopping times, we refer the reader to
\cite{nikeghbaliyor05}.\\

Theorem \ref{riskMininStrategy} and Corollary \ref{riskMininStrategyPredictable} give the general relation between the
 \(\mathbb G\)-risk-minimizing strategy in the model \((S^{\tau},\mathbb G)\) for the claim \(h_\tau\) at term \(T\) and the
  \(\mathbb F\)-risk-minimizing strategy in \((S,\mathbb F)\) for the claim
   \(\E[\int_0^{\infty}h_u dD^{o,\mathbb F}_u\ \big|\ {\cal F}_T]\) (or $\E\left[ \int_0^\infty h_u d F_u |  {\cal F}_T \right]$ when \(h\) is \(\mathbb F\)-predictable) at term \(T\).
In the next subsections,  we further establish the arising \(\mathbb F\)-risk-minimizing strategies for certain specific mortality contracts to highlight the impact of the interplay/correlation between the financial market and the death time. \\

Corollary \ref{RiskMiniPseudo} extends already the results of \cite[Chapter 5]{barbarin09}  and \cite{biaginibotero15, biaginirheinlander16, biaginischreiber13} to more broader models  of $(S,\mathbb F, \tau)$ and to a larger class of mortality liabilities. Indeed, in these papers, the authors study risk-minimization of life insurance liabilities consisting of two main building blocks: pure endowment and term insurance only, while an annuity contract can be dealt with as a combination of both. This literature derives the results, by applying the hazard rate approach of credit derivatives (see, e.g., \cite{bieleckirutkowski02}), under the following several assumptions:
\begin{enumerate}
\item[(a)] The random time $\tau$ is assumed to avoid $\mathbb F$-stopping times. This allows  $\tau$ to be a totally inaccessible $\mathbb G$-stopping time, and 
      $\Delta U_{\tau}=0$ for any $\mathbb F$-adapted RCLL process $U$.
\item[(b)]  The process  $G$ is strictly positive, i.e., the stopping time  $R=+\infty$ $P$-a.s..
\item[(c)]  The payment processes and payoff processes are predictable processes.
\item[(d)]  The H-hypothesis holds (i.e., $M^{\tau}$ is a $\mathbb G$-local martingale for any $\mathbb F$-local martingale $M$). 
\end{enumerate}
This H-hypothesis is relaxed in \cite{barbarin09}, while the author could not get rid of assumptions (a), (b), and (c), as his approach relies essentially on the martingale decomposition of \cite{blanchetjeanblanc04}. \\
 
%%%%%%%%%%XXXXXXXXXXXXXXXXXXXXXXXXXXXXXXXXXXXX
%%%%%%%% SUBSECTION 4.3
%%%%%%%%%%XXXXXXXXXXXXXXXXXXXXXXXXXXXXXXXXXXX

%\subsection{Practical cases for the payoff process $h$} \label{secRiskMinimisationCorrelation}
Throughout  the rest of the paper, we consider the following survival probabilities
\begin{equation}\label{G(T)process}
F_t (s) := P(\tau \le s | {\cal F}_t) , \ \ \ \ \ \mbox{and}\ \ \ \ \ \
G_t (s):= P(\tau > s | {\cal F}_t)=1-F_t(s),\ \ \ \ \ \forall\ \ s,t \in [0,T],
\end{equation}
and the following mortality/longevity derivatives. 

\begin{definition}\label{insuranceConctracts} Consider $T\in (0,+\infty)$, $g\in L^1({\cal F}_T)$ and $K\in L^1({\cal O}(\mathbb F), P\otimes D)$. \\
{\rm{(a)}} A {zero-coupon} longevity bond is an insurance contract that pays the conditional survival probability at term $T$ (i.e., an
insurance contract with payoff $G_T=P(\tau>T|{\cal F}_T)$). \\
{\rm{(b)}} A pure endowment insurance, with benefit $g$, is an insurance contract that pays $g$ at term $T$ if the insured survives
(i.e., an insurance contract with payoff  $gI_{\{\tau>T\}}$).\\
%%{\rm{(c)}} A term insurance contract with benefit process $K$ is an insurance contract that pays $K_{\tau}$ at $\tau$ 
%%if the insured dies before or at the term of the contract (i.e., an insurance contract with payoff $K_{\tau}I_{\{\tau\leq T\}}$).\\
{\rm{(c)}} An endowment insurance contract with benefit pair $(g,K)$, is an insurance contract that pays $g$ at term $T$ if the insured survives and pays $K_{\tau}$ at the time of death if the insured dies before or at the maturity (i.e,. its payoff is $gI_{\{\tau>T\}}+K_{\tau}I_{\{\tau\leq T\}}$).
\end{definition}

\subsection{The case of pure endowment insurance contract}\label{SubsetionpayoffPrevisible}

This subsection considers the case of a pure endowment contract with benefit  $g$ defined in Definition \ref{insuranceConctracts}-(b). Thus, the payoff process for this contract takes the form of
\begin{equation}\label{h1}
h_t := g I_{\Rbrack T, +\infty\Lbrack} (t),\ \ \ \ \ \ \ g\in L^2( {\cal F}_T, P).
\end{equation}
The following describes precisely the risk-minimizing 
strategy in terms of the risk-minimizing strategies for the financial, mortality, and correlation components.

\begin{theorem} \label{PropAtTerm}
Suppose that (\ref{mainassumpOn(X,tau)}) holds, and consider $h$ given by (\ref{h1}).  Then the following hold.\\
{\rm{(a)}} The risk-minimizing strategy for the mortality claim \(h_{\tau}\), under $(S^{\tau},\mathbb G)$, takes the form of
\begin{equation}\label{StratAtTerm}
\xi^{(h,\mathbb G)}:=\left(G_{-}(T) \xi^{(g,\mathbb F)} + U_-^g \xi^{(G_T,\mathbb F)} +
 \xi^{({\rm Cor}_T,\mathbb F)}\right)\big(G_{-}+\varphi^{(m)}\big)^{-1}I_{\Rbrack 0,\tau\Rbrack},
 \end{equation}
and the corresponding remaining risk is given by
\begin{align}
L^{(h,\mathbb G)} &:= - {\frac{ G_{-}(T) \xi^{(g,\mathbb F)} + U_-^g \xi^{(G_T,\mathbb F)} + \xi^{({\rm Cor}_T,\mathbb F)}}{G_{-}
\big(G_{-}+\varphi^{(m)}\big)}}I_{\Rbrack 0,\tau\Rbrack}\is\widehat{L^{(m)}}
 +I_{\Rbrack 0,\tau\Rbrack} {\frac{G_{-}(T)}{G_{-}}} \is\widehat{L^{(g,\mathbb F)}}\nonumber\\
& \quad + I_{\Rbrack 0,\tau\Rbrack}{\frac{ U_-^g}{G_{-}}} \is\widehat{L^{(G_T,\mathbb F)}}  +
{\frac{ I_{\Rbrack 0,\tau\Rbrack}}{G_{-}}} \is\widehat{L^{({\rm Cor}_T,\mathbb F)}}{
-\frac{M^{(g)}_{-}}{ G^2_{-}} I_{\Rbrack 0,\tau\wedge T\Rbrack} \is\widehat{m}    
-\frac{M^{(g)}}{ G} I_{\Rbrack 0,R\Lbrack} \is \left(N^{\mathbb G}\right)^T.}
\label{StratAtTermRemain}
\end{align}
Here, the {\it correlation process} ${\rm Cor}=({\rm Cor}_t)_{t\geq 0}$ is given by
\begin{equation}\label{Rho4(g,tau)}
{\rm Cor}_t:=[G(T),U^g]_t+ {\Cov}\Bigl(I_{\{\tau > T\}}, {g}\ \big|\ {\cal F}_t\Bigr),
\end{equation}
and  $U^g$ and $M^{(g)}$ are two  $\mathbb{F}$-martingales given by $U^g_t=\E [ g\mid \mathcal{F}_t]$
 and $M^{(g)}_t=\E [ gG_T\mid \mathcal{F}_t]$ respectively.\\
The pairs
\(\left(\xi^{(g,\mathbb F)},L^{(g,\mathbb F)}\right) \), \(\left(\xi^{(G_T,\mathbb F)}, L^{(G_T,\mathbb F)}\right)\),
 and \(\left(\xi^{({\rm Cor}_T,\mathbb F)}, L^{({\rm Cor}_T,\mathbb F)}\right)\) are the risk-minimizing strategies
  and the remaining risks, under $(S,\mathbb F)$, for the claims \({g}\), \(G_T\), and \({\rm Cor}_T\) respectively, while $(\varphi^{(m)},L^{(m)})$ is defined in (\ref{decompositionUandhatU}).\\
{\rm{(b)}} The value process, \(V(\rho^{*,\mathbb G})\), of the risk-minimizing portfolio under the model $(S^{\tau},\mathbb G)$, is given by
\begin{equation} \label{ValueAtTerm}
V(\rho^{*,\mathbb G})=G^{-1}\ ^{o,\mathbb F}\!\left(h_{\tau}I_{\Lbrack 0,\tau\Lbrack}\right)I_{\Lbrack 0,\tau\wedge T\Lbrack}.
\end{equation}
\end{theorem}

\noindent The amount $g$ of a pure endowment is purely financial. The \(\mathbb F\)-strategy and the remaining risk for the claim
 \(\E\Bigl[ \int_0^\infty h_u dF_u\ \big|\ {\cal F}_T \Bigr]={g} (1-F_T) = {g} G_T\) are expressed as functions of  the corresponding strategy
  and risk for this  pure financial claim \({g}\), for the pure mortality claim $G_T$ and the correlation ${\rm Cor}_T$ between
   the pure financial market and the mortality model including the time of death. \\
   When $g$ is deterministic then $M^{(g)}=gG(T)$, the martingale $U^{(g)}$ is equal to $g$, and the correlation process $({\rm Cor}_t)_{0\leq t\leq T}$ is a null process. 
   Thus, we get $(\xi^{(g,\mathbb F)},L^{(g,\mathbb F)})=(\xi^{({\rm Cor}_T,\mathbb F)},L^{({\rm Cor}_T,\mathbb F)})\equiv (0,0)$, 
   and  conclude that in this case the pair $(\xi^{(h,\mathbb G)},L^{(h,\mathbb G)}) $ takes the following form:
	\begin{align*}
\xi^{(h,\mathbb G)} &:=g \xi^{(G_T,\mathbb F)} \big(G_{-}+\varphi^{(m)}\big)^{-1}I_{\Rbrack 0,\tau\Rbrack},\\
L^{(h,\mathbb G)} &:= - {\frac{g \xi^{(G_T,\mathbb F)}I_{\Rbrack 0,\tau\Rbrack}}{G_{-}
		\big(G_{-}+\varphi^{(m)}\big)}}\is\widehat{L^{(m)}}
 +{\frac{ g I_{\Rbrack 0,\tau\Rbrack}}{G_{-}}} \is\widehat{L^{(G_T,\mathbb F)}}
	-\frac{g G_{-}(T)}{ G^2_{-}} I_{\Rbrack 0,\tau\wedge T\Rbrack} \is\widehat{m}
	-\frac{g G(T)}{ G} I_{\Rbrack 0,R\Lbrack} \is\left(N^{\mathbb G}\right)^T.
	\end{align*}
%%%%%%%%%%%%XXXXXXXXXX
%%%%%%%%%%%%%%%%%%%%%%%%%%%%%%%
This remaining risk $L^{(h,\mathbb G)} $ contains integrals with respect to $N^{\mathbb G}$ and $\widehat{m}$ that 
represent the unsystematic component of the mortality risk and a combination of systematic and unsystematic mortality risk respectively. \\
For this particular case of a pure endowment contract we further compare the pair
$\left(\xi^{(h,\mathbb F)}, L^{(h,\mathbb F)}\right)$ obtained in \cite{barbarin09, biaginibotero15, biaginirheinlander16,biaginischreiber13} 
with our pair given by \eqref{StratAtTerm}-\eqref{StratAtTermRemain}.  In \cite{barbarin09}, the author assumes that the financial market is independent of the mortality model in the
sense that the process $\rm{Cor}$, defined in \eqref{Rho4(g,tau)}, is null. Further, he assumes that $G(T)$ is strongly orthogonal to $S$ meaning that the systematic risk
mortality component cannot be hedged by investing in $S$. This implies that in \eqref{XiFdecomposed}
$\left(\xi^{(G_T,\mathbb F) }, L^{(G_T,\mathbb F) }\right)=(0, G(T))$.  In \cite{biaginibotero15, biaginirheinlander16,biaginischreiber13} , it is also assumed that $G(T)$ is driven by a
local $\mathbb F$-martingale $Y$ which is strongly orthogonal to  $S$  but follow a slightly different approach. They construct a predictable
decomposition of $M^{(g)}$ in terms of $S$ and $Y$ instead of the expression \eqref{decompositionmg}. Hence, the authors do not distinguish the three components (pure financial, pure mortality and correlation) as we do. In \cite[Chapter 5]{barbarin09} , the author studies risk-minimization for a pure endowment contract  under strict independence between the financial and insurance market and other further assumptions on $(S,\mathbb F,\tau)$. It is clear that by imposing additional specifications such as $\tau$ is an $\mathbb F$-pseudo stopping time, or $\tau$ avoids all $\mathbb F$-stopping times, or all $\mathbb F$-martingales are continuous,  our result boils down to that of  in \cite[Proposition 5.1]{barbarin09}. Therefore we conclude that Theorem \ref{PropAtTerm} generalizes \cite{barbarin09}  in several directions.

\begin{proof}[Proof of Theorem \ref{PropAtTerm}]: 
Notice that for the payoff process  $h$ given in (\ref{h1}), we have $h\equiv 0$ on $[0,T]$,   $\E[ \int_0^\infty h_u dD^{o,\mathbb F}_u \big| {\cal F}_T ] =gG_T$ and 
\begin{align*}
M^h_t=M_t^{(g)} &:=
 \ ^{o,\mathbb  F}\left( G_T {g} \right)_t
= \E[I_{\{\tau>T\}}| {\cal F}_t] \E\left[ {g} \big| {\cal F}_t\right]
 + {\Cov} \left(I_{\{\tau>T\}} , {g} \big| {\cal F}_t\right):= G_t(T) U_t^g + {\Cov}_t^g .
\end{align*}
Therefore, in virtue of Corollary  \ref{riskMininStrategyPredictable} (see also Theorem \ref{riskMininStrategy}) 
the proof of Theorem \ref{PropAtTerm} follows immediately as soon as we prove that
\begin{align}\label{XiFdecomposed}
	\xi^{(h,\mathbb F)} &=  G_{-}(T) \xi^{(g,\mathbb F)} + U_-^g \xi^{(G_T,\mathbb F)} +
\xi^{({\rm Cor}_T,\mathbb F)}\ \ \mbox{ and}\nonumber\\
	L^{(h,\mathbb F)}&= G_{-}(T) \is L^{(g,\mathbb F)}  + U_-^g \is L^{(G_T,\mathbb F)} + L^{({\rm Cor}_T,\mathbb F)}.
\end{align}
A direct application of the integration by parts formula to $G_t(T) U_t^g$ leads to
\begin{equation} \label{decompositionmg}
{M^{(g)}} =  G_0(T) U_0^g + G_{-}(T) \is U^g + U_-^g \is G(T) + \mbox{Cor},
\end{equation}
where $\mbox{Cor}$ is the process defined in (\ref{Rho4(g,tau)}). In order to apply the GKW decomposition for each
 of the $\mathbb F$-local martingale in the RHS term of (\ref{decompositionmg}),
  we need to prove that these local martingale are actually (locally) square integrable martingales.
   To this end, we remark that $0\leq G_{-}(T)\leq 1$ and $U^g$ is a square integrable $\mathbb F$-martingale.
    Furthermore, we derive $\displaystyle\sup_{0\leq t\leq T}\vert
     {M_t^{(g)}}\vert\leq \sup_{0\leq t\leq T}\E\left[\ \vert g\vert\ \big|\ {\cal F}_t\right]\in L^2(\Omega, {\cal F},P)$ and
 \begin{align*}
 	\E\Bigl[(U^g_{-})^2\is [G(T)]_T\Bigr]&\leq  \E\left[\int_0^T \sup_{0\leq s< t}(U^g_s)^2 d[G(T)]_t\right]=
 	\E\left[\int_0^T ([G(T)]_T-[G(T)]_t)\,d\sup_{0\leq s\leq t}(U^g_s)^2\right]\\
 	&=
 	\E\left[\int_0^T \E ([G(T)]_T-[G(T)]_t\ \big|\ {\cal F}_t)\, d\sup_{0\leq s\leq t}(U^g_s)^2\right]
 \leq \E\left[\sup_{0\leq s\leq T}(U^g_s)^2\right]<+\infty.
 \end{align*}
   As a result, the three local martingale ${M^{(g)}}$, $G(T)_{-}\is U^g$ and $U^g_{-}\is G(T)$ are square
   integrable martingales, and subsequently $G(T)_{-}\is U^g$, $U^g_{-}\is G(T)$ and $\mbox{Cor}$
    are square integrable martingales.\\
 Therefore, by applying the GKW decomposition to $U^g$, $G(T)$ and $\mbox{Cor}$, we obtain
\[
M^{(g)} =  M_0^{(g)} + \left(G_{-}(T) \xi^{(g,\mathbb F)} + U_-^g \xi^{(G_T,\mathbb F)} + \xi^{({\rm Cor}_T,\mathbb F)} \right)\is S
  + G_{-}(T) \is L^{(g,\mathbb F)}  + U_-^g \is L^{(G_T,\mathbb F)} + L^{({\rm Cor}_T,\mathbb F)},
\]
and the proof of (\ref{XiFdecomposed}) follows immediately. This ends the proof of assertion (a).\\
 Concerning the value process of the corresponding portfolio, we note that the payment process \(A\) is given by
 $A_t=I_{\{t=T\}}I_{\{\tau > T\}}g=I_{\{t=T\}}h_{\tau}$ such that $A_T-A_t=(1-I_{\{t=T\}})h_{\tau}$ and
\[
V_t(\rho^{*,\mathbb G}) = (1-I_{\{t=T\}})H_t,
\]
with $H$ given by \eqref{DecompHexplicit} where the first term is zero since we do not hedge beyond the term of the contract,
 thus $I_{\{\tau>T\}}I_{\{\tau\leq t\}}=0$. This ends the proof of the theorem.\end{proof}
%%%%%%%%%%%%%XXXXXXXXXXXX
%%%%%%%%%%%%%%%%%%%%%%%

\begin{corollary}\label{cases4correlations} Consider the mortality claim $h_{\tau}$ where $h$ is given by (\ref{h1}),
 and the square integrable $\mathbb F$-martingale $U^g_t:=\E[g\ \big|\ {\cal F}_t]$. Then the following assertions hold.\\
{\rm{(a)}} Suppose that $\tau$ is a pseudo-stopping time. Then the pair $(\xi^{(h,\mathbb G)},L^{(h,\mathbb G)})$, of (\ref{StratAtTerm})-(\ref{StratAtTermRemain}), becomes
\begin{align}
\xi^{(h,\mathbb G)}&:= \Bigl({\frac{G_{-}(T)}{G_{-}}} \xi^{(g,\mathbb F)} + {\frac{U_-^g}{G_{-}}} \xi^{(G_T,\mathbb F)}
 +{\frac{1}{G_{-}}} \xi^{({\rm Cor}_T,\mathbb F)}\Bigr)I_{\Rbrack 0,\tau\Rbrack},\label{stategyCase2}\\
L^{(h,\mathbb G)} &:=  {\frac{G_{-}(T)I_{\Rbrack 0,\tau\Rbrack}}{G_{-}}} \is L^{(g,\mathbb F)}
+  I_{\Rbrack 0,\tau\Rbrack} {\frac{U_-^g}{G_{-}}} \is L^{(G_T,\mathbb F)}
+  {\frac{I_{\Rbrack 0,\tau\Rbrack}}{G_{-}}} \is L^{({\rm Cor}_T,\mathbb F)}
-\frac{M^{(g)}}{ G} I_{\Rbrack 0,R\Lbrack} \is \left(N^{\mathbb G}\right)^T.\label{L(h,G)process2}
\end{align}
 {\rm{(b)}}  Suppose $\tau$ is independent of ${\cal F}_{\infty}$ and $P(\tau>T)>0$.
 Then $(\xi^{(h,\mathbb G)},L^{(h,\mathbb G)})$ takes the following form 
\begin{equation}
\xi^{(h,\mathbb G)}_t:={\frac{P(\tau>T)}{P(\tau\geq t)}} \xi^{(g,\mathbb F)}_tI_{\{t\leq \tau\}},\ \ \
L^{(h,\mathbb G)}_t :=\int_0^{t\wedge\tau} {\frac{P(\tau>T)}{P(\tau\geq s)}} dL^{(g,\mathbb F)}_s
 -\int_0^{t\wedge T} {\frac{P(\tau>T)U^{(g)}_s}{P(\tau>s)}}dN^{\mathbb G}_s.\label{L(h,G)process1}
\end{equation}
\end{corollary}

\begin{proof} It is clear that, when $\tau$ is independent of ${\cal F}_{\infty}$, we have
$$
G_t(T)=P(\tau>T)=G_{t-}(T),\ \ \ G_t=P(\tau>t),\ \ \ G_{t-}=P(\tau\geq t)=\widetilde G_t,\ \ \ m\equiv m_0,\ \ \ {\rm Cor}\equiv 0.$$
As a consequence, $\tau$ is a pseudo-stopping time  and  
\begin{equation*}
\xi^{(G_T,\mathbb F)}\equiv  0,\ \ \ L^{(G_T,\mathbb F)}\equiv 0,\ \ \xi^{({\rm Cor}_T,\mathbb F)}\equiv 0,
\ \ \ L^{({\rm Cor}_T,\mathbb F)}\equiv 0.
\end{equation*}
Thus, by plugging these in (\ref{stategyCase2}) and (\ref{L(h,G)process2}) and using the facts that $R>T$ $P$-a.s. (due to the assumption $G_T=P(\tau>T)>0$) and $M^{(g)}_t:=\E[gG_T|{\cal F}_t]=P(\tau>T)U^{(g)}_t$, assertion (b) follows immediately from assertion (a). Hence, the rest of the proof focuses on proving assertion (a).  To this end, recall that when $\tau$ is a pseudo-stopping time, we have $m\equiv m_0$, and as a consequence we get
$$\varphi^{(m)}\equiv 0,\quad L^{(m)}\equiv 0, \quad {\rm{and}}\quad  \widehat{M}=M^{\tau}\quad\mbox{for any}\quad M\in{\cal M}_{loc}(\mathbb F). $$
Hence, by inserting these in (\ref{StratAtTerm}) and (\ref{StratAtTermRemain}), assertions (a) follows immediately, and the proof of the corollary is completed.\end{proof}

\subsection{The case of annuity up-to the time of death}
Herein, we address an annuity paid until the time of death of the policyholder, or until the end of the contract $T$.
 This insurance contract is also called endowment insurance and is defined more generally in Definition \ref{insuranceConctracts}-(c).  Let \(C:=(C_t)_{t\geq 0}\) be  the
 $\mathbb F$-optional and square integrable (with respect to $P\otimes D$) process such that $C_t$ represents
 the discounted accumulated amount up to time $t$ paid by the insurer, with \(C_0=0\).
   Then,  \(I_{\{\tau > T\}}{C}_T + I_{\{\tau\le T\}} {C}_\tau \) gives the discounted payoff up to the time of death or the end $T$ of
   the contract  whatever occurs first. Thus, the payoff process $h$ takes the form
of \begin{equation}\label{h3}
h_t:=I_{\{t > T\}}{C}_T + I_{\{t\le T\}}{C}_t.
\end{equation}

\begin{theorem} \label{PropAnnuity}
Suppose that (\ref{mainassumpOn(X,tau)}) holds, and $h$ is given by (\ref{h3}). Let \(U^{K}\) be the $\mathbb F$-martingale
 $U^{K}_t :=\E [K \mid {\cal F}_t]$ for any $K\in L^2({\cal F}_T,P)$. Then the following assertions hold.\\
 {\rm{(a)}} The risk-minimizing strategy and the remaining risk for the mortality claim
  \(h_{\tau}\), under the model $(S^{\tau},\mathbb G)$, are  given by
\begin{eqnarray}\label{StrategyforAnuity}
\xi^{(h,\mathbb G)}:={\frac{G_{-}(T) \xi^{({C_T},\mathbb F)} + U_-^{C_T} \xi^{(G_T,\mathbb F)} + \xi^{({\rm Cor}_T,\mathbb F)}
+\xi^{({{\widetilde C}_T},\mathbb F)}}{G_{-}+\varphi^{(m)}}}I_{\Rbrack 0,\tau\Rbrack}.
\end{eqnarray}
and
\begin{align}\label{RemainRisk4Anuity}
L^{(h,\mathbb G)} :=\ &- {\frac{\xi^{({{\widetilde C}_T},\mathbb F)}+ G_{-}(T) \xi^{({C_T},\mathbb F)}
+ U_-^{C_T} \xi^{(G_T,\mathbb F)} + \xi^{({\rm Cor}_T,\mathbb F)}}{G_{-}
\big(G_{-}+\varphi^{(m)}\big)}}I_{\Rbrack 0,\tau\Rbrack}\is\widehat{L^{(m)}}\nonumber\\
 &+I_{\Rbrack 0,\tau\Rbrack} {\frac{G_{-}(T)}{G_{-}}} \is\widehat{L^{({C_T}},\mathbb F)}
  +{\frac{I_{\Rbrack 0,\tau\Rbrack}}{G_{-}}} \is \widehat{L^{({{\widetilde C}_T}},\mathbb F)}
   + I_{\Rbrack 0,\tau\Rbrack}{\frac{ U_-^{C_T}}{G_{-}}} \is\widehat{L^{(G_T,\mathbb F)}} \nonumber\\
& + {\frac{I_{\Rbrack 0,\tau\Rbrack}}{G_{-}}} \is\widehat{L^{({\rm Cor}_T,\mathbb F)}}
-{\frac{M^{(C_T)}_{-}}{ G^2_{-}}}I_{\Rbrack 0,\tau\wedge T\Rbrack} \is\widehat{m}
 -{\frac{M^{(C_T)}}{G}} I_{\Lbrack 0,R\Lbrack} \is \left(N^{\mathbb G}\right)^T.
\end{align}

Herein, $M_t^{(C_T)}:=\E [C_TG_T \mid {\cal F}_t]$,
\begin{eqnarray}\label{C(11,2)processes}
{\rm Cor}_t :=[G(T),U^{C_T}]_t + {\Cov} \left(I_{\{\tau> T\}}, {C}_T \big| {\cal F}_t \right),\quad \mbox{and}\quad
{\widetilde C}_t :=\int_0^t C_u dD^{o,\mathbb F}_u.
\end{eqnarray}
The pairs of processes $\left(\xi^{({C_T},\mathbb F)}, L^{({C_T},\mathbb F)}\right)$, $\left(\xi^{(G_T,\mathbb F)},L^{(G_T,\mathbb F)}\right)$,
 $\left(\xi^{({\rm Cor}_T,\mathbb F)},L^{({\rm Cor}_T,\mathbb F)}\right) $, and \\
  $\left(\xi^{({\widetilde C}_T,\mathbb F)},L^{({\widetilde C}_T,\mathbb F)}\right) $ are the risk-minimizing strategies and
   the remaining (undiversified) risk, under the model $(S,\mathbb F)$, for the contracts with claims \({C}_T\), \(G_T\), \({\rm Cor}_T\),
    and \({\widetilde C}_T\) respectively, and the pair $(\varphi^{(m)},L^{(m)})$ is given by
     Lemma \ref{consequences4MainAssum}.\\
{\rm{(b)}} The value of the risk-minimizing portfolio \(V(\rho^{*,\mathbb G})\) under the model $(S^{\tau},\mathbb G)$, is given by	
\begin{eqnarray}
V(\rho^{*,\mathbb G})=  G^{-1}\ ^{o,\mathbb F}\!\left(h_{\tau}I_{\Lbrack 0,\tau\Lbrack}\right)I_{\Lbrack 0,\tau\Lbrack}-
 I_{[T]} G^{-1}\ ^{o,\mathbb F}\!\left(I_{\{\tau\le T\}} {C}_\tau I_{\Lbrack 0,\tau\Lbrack}\right)I_{\Lbrack 0,\tau\Lbrack}.
\end{eqnarray}
\end{theorem}

\begin{proof} Thanks to Theorem \ref{riskMininStrategy}, the above theorem will follow immediately as long as we prove that
\begin{align}\label{XihandLh}
\xi^{(h,\mathbb F)}&= G_{-}(T) \xi^{({C_T},\mathbb F)} + U_-^{C_T} \xi^{(G_T,\mathbb F)} + \xi^{({\rm Cor}_T(C_T),\mathbb F)}
+\xi^{(U^{{\widetilde C}_T},\mathbb F)},\nonumber\\
L^{(h,\mathbb F)}&= G_{-}(T)\is{L^{({C_T},\mathbb F)}} +{L^{({{\widetilde C}_T},\mathbb F)}}+
 U_-^{C_T} \is {L^{(G_T,\mathbb F)}}+ {L^{({\rm Cor}_T,\mathbb F)}},
\end{align}
where $\left(\xi^{(h,\mathbb F)},L^{(h,\mathbb F)}\right)$ is the risk-minimizing strategy and the remaining (undiversified) risk of the payoff
$\E\left[\int_0^{\infty}h_u dD^{o,\mathbb F}_u\ \big|\ {\cal F}_T\right]$ under the model $(S,\mathbb F)$. To prove (\ref{XihandLh}),
 we first remark that $h=h^{(1)}+h^{(2)}$, where $h^{(1)}$ --has the same form as the payoff process of
  Subsubsection \ref{SubsetionpayoffPrevisible}-- and $h^{(2)}$ are given by
\begin{equation}\label{H1H2}
h^{(1)}_t:=C_T I_{\{t>T\}}\ \ \ \ \mbox{and}\ \ \ \  h^{(2)}_t=I_{\{t\leq T\}}C_t.
\end{equation}
Thus, we derive
$$ \E\left[\int_0^{\infty}h_u dD^{o,\mathbb F}_u\ \big|\ {\cal F}_T\right]=
\E\left[\int_0^{\infty}h_u^{(1)} dD^{o,\mathbb F}_u\ \big|\ {\cal F}_T\right]+\int_0^T C_u dD^{o,\mathbb F}_u
=:\E\left[\int_0^{\infty}h_u^{(1)} dD^{o,\mathbb F}_u\ \big|\ {\cal F}_T\right]+{\widetilde  C}_T,$$
and deduce that
\begin{equation*}
\xi^{(h,\mathbb F)}=\xi^{(h^{(1)},\mathbb F)}+\xi^{({\widetilde C}_T,\mathbb F)},\quad \mbox{and}\quad L^{(h,\mathbb F)}
=L^{(h^{(1)},\mathbb F)}+L^{({\widetilde C}_T,\mathbb F)}.
\end{equation*}
Therefore, by combining this with Theorem \ref{PropAtTerm} with $g=C_T$, the proof of (\ref{XihandLh}) follows immediately.\\
The value process $V(\rho^{*,\mathbb G})$ of the risk-minimizing strategy under the model $(S^{\tau},\mathbb G)$ also consists of two
parts given by \eqref{ValueAtTerm} and \eqref{ValueTI} for $h^{(1)}$ and $h^{(2)}$, respectively. This ends the proof of the theorem.
\end{proof}

%%%XXXXXXXXXXXXXXXXXXXXXXX TO BE RE-READ.....RE-CHECK.....
Similarly as in the case of a pure endowment, for the annuity contract, in  \cite{biaginischreiber13}  the authors derive a predictable decomposition for the martingale $M^{(C_T)}$  in terms of $S$ and of the $\mathbb F$-local
martingale $Y$ which drives $G(T)$ and which is strongly orthogonal to $S$.  Hence, their assumptions correspond to $\rm{Cor}=0$,  $\xi^{(G_T,\mathbb F) }=0$ and there is a
term in $Y$ instead of our $L^{(G_T,\mathbb F) }$ in \eqref{XihandLh}. In \cite{barbarin09}, the author proceeded very differently and  did not write the payoff of the annuity contract as a sum of a pure endowment and a term insurance, while he worked with the integral expression. This made the results very involved and hard to interpret. Again, for the annuity contract, \cite{barbarin09}  falls in the case where  $\rm{Cor}=0$  and $G(T)$ is orthogonal to $S$ (i.e., $\langle G(T), S\rangle^{\mathbb F}\equiv 0$).
%%%%XXXXXXXXXXXXXXXXXXXXXXXXXXXXXX

\begin{corollary}\label{cases4Anuity} Consider the mortality claim $h_{\tau}$ where $h$ is given by (\ref{h3}),
 and put $U^K_t:=\E[K\ \big|\ {\cal F}_t]$ and $M^{(K)}_t:=\E[G_T K|{\cal F}_t]$  for any $K\in L^2({\cal F}_T,P)$.
  Then the following assertions hold.\\
{\rm{(a)}} Suppose $\tau$ is a pseudo-stopping time. Then the pair $(\xi^{(h,\mathbb G)},L^{(h,\mathbb G)})$, of (\ref{StrategyforAnuity})-(\ref{RemainRisk4Anuity}), becomes
\begin{align}
\xi^{(h,\mathbb G)}:=\ &\Bigl({\frac{G_{-}(T)}{G_{-}}} \xi^{(C_T,\mathbb F)} +
 {\frac{U_-^{C_T}}{G_{-}}} \xi^{(G_T,\mathbb F)} +{\frac{1}{G_{-}}} \xi^{({\rm Cor}_T,\mathbb F)}+
 {\frac{1}{G_{-}}} \xi^{(\widetilde{C}_T,\mathbb F)}\Bigr)I_{\Rbrack 0,\tau\Rbrack},\label{stategyCase22}\\
L^{(h,\mathbb G)} :=\ & I_{\Rbrack 0,\tau\Rbrack} {\frac{G_{-}(T)}{G_{-}}} \is L^{(C_T,\mathbb F)}+
{\frac{I_{\Rbrack 0,\tau\Rbrack}}{G_{-}}} \is L^{({{\widetilde C}_T},\mathbb F)}  +
 I_{\Rbrack 0,\tau\Rbrack} {\frac{U_-^{C_T}}{G_{-}}} \is L^{(G_T,\mathbb F)}\nonumber\\
&  + G_{-}^{-1}I_{\Rbrack 0,\tau\Rbrack} \is L^{({\rm Cor}_T,\mathbb F)}
 -{\frac{M^{(C_T)}}{G}} I_{\Lbrack 0,R\Lbrack} \is\left( N^{\mathbb G}\right)^T.\label{L(h,G)process22}
\end{align}
{\rm{(b)}}  Suppose that $\tau$ is independent of the initial market ${\cal F}_{\infty}$, and $P(\tau>T)>0$.
 Then we get 
 \begin{align}
\xi^{(h,\mathbb G)}_t&:= {\frac{P(\tau>T)\xi^{(C_T,\mathbb F)}_t+
\xi^{(\widetilde{C}_T,\mathbb F)}_t}{P(\tau\geq t)}} I_{\{t\leq \tau\}},\label{L(h,G)process11}\\
L^{(h,\mathbb G)}_t &:= \int_0^{t\wedge\tau} {\frac{P(\tau>T)}{P(\tau\geq s)}} dL^{(C_T,\mathbb F)}_s+
\int_0^t{\frac{1}{P(\tau>s)}} dL^{({\widetilde C}_T,\mathbb F)}_s	-\int_0^{t\wedge T} {\frac{P(\tau>T)}{P(\tau>s)}}dN^{\mathbb G}_s
.\end{align}
\end{corollary}

\begin{proof} The proof of this corollary mimics the proof of  Corollary \ref{cases4correlations} using Theorem \ref{PropAnnuity} instead.
\end{proof}
%%%%%%%%%%%%%%%%%%%%%%%%%%%%%%%%%%%%%%%%
%%%%%%%%%% SUBSECTION 4.4
%%%%%%%%%%%%%%%%%%%%%%%%%%%%%%%%%%%%%%%%%%%%%%%%%%%%%%%

%%%\subsection{Proofs of Theorems \ref{riskMininStrategy} and \ref{PropAtTerm}}\label{Subsection4proofs}
%%Herein, we prove these theorems which represent the main results of Subsections \ref{Subsection4results} and \ref{SubsetionpayoffPrevisible}. 

%%%%%%%%%%%%%%%%%%%%%%%%%%%%%%%%%%%%%%%%%%%%%%%%%%%
%%%%%%%%%%%%%%
%%%%%%%%%%%%%% SECTION 4 SECURITIZATION
%%%%%%%%%%%%
%%%%%%%%%%%%%%%%%%%%%%%%%%%%%%%%%%%%%%%%%%%%%%%%%%%%%

\section{Hedging mortality risk with insurance securitization}\label{SectionSecuritisation}

In this section, we address the hedging problem for mortality liabilities, using the risk-minimization criterion of Subsection \ref{subsection4Riskmini}, by investing in both the stock and one (or more) of the morality/longevity securities defined in Definition \ref{insuranceConctracts} . This process is known, in the insurance literature, as insurance securitization. Thus, we need to specify the dynamics of these derivatives that will be used in the securitization process.

The following theorem, see  \cite[Theorem 3.1]{choullidaveloose2018}, elaborates the prices' dynamics of these insurance contracts while allowing the death time $\tau$ to have an arbitrary model and under the assumption that {\it $P$ is
 a risk-neutral probability for the model
$(\Omega,\mathbb G)$}. {\bf This assumption means that all discounted price processes of traded securities in the market $(\Omega,\mathbb G)$ are local martingales under $P$}. In our view this assumption on the probability $P$  is not a restriction. In fact, on the one hand, it is dictated by the F\"ollmer-Sondermann optimization method of Subsection  \ref{subsection4Riskmini} considered herein. On the other hand, one can calculate every process used in the theorem (starting with the processes $G$, $\widetilde G$, $G_{-}$) under a chosen risk-neutral measure, $Q$, for the informational model $(\Omega,\mathbb G)$. It is important to mention that the following theorem is a direct application of our martingale representation theorems \cite[Theorems 2.19 and  2.22]{choullidaveloose2018}.

\begin{theorem}\cite[Theorem 3.1]{choullidaveloose2018}\label{contractsStructures} Suppose $P$ is a risk-neutral probability for $(\Omega, \mathbb G)$. Then the following hold.\\
{\rm{(a)}} The discounted price process of the pure endowment insurance contract with benefit $g\in L^1({\cal F}_T, P)$  at term $T$, is denoted by
$P^{(g)}$, and is given by
\begin{equation}\label{PureEndowment}
P^{(g)}=P^{(g)}_0+{\frac{I_{\Rbrack 0,\tau\wedge T\Rbrack}}{G_{-}}}\is\widehat{M^{(g)}}
-{\frac{M^{(g)}_{-}}{G_{-}^2}}I_{\Rbrack 0,\tau\wedge T\Rbrack}\is\widehat m -{\frac{ M^{(g)}}{G}}I_{\Rbrack 0,R\Lbrack}\is (N^{\mathbb G})^T,
\  \mbox{with}\ M^{(g)}_t:=\E\left[gG_T\ \big|\ {\cal F}_t\right].
\end{equation}
{\rm{(b)}} The discounted price process of the longevity bond, with term $T$, is denoted by $B$ and satisfies 
\begin{align}\label{LongevityBond}
B^{\tau}&= B_0+{\frac{I_{\Rbrack 0,\tau\wedge T\Rbrack}}{G_{-}}}\is \widehat{M^{(B)}}
-{\frac{M^{(B)}_{-}-{\overline D}^{o,\mathbb F}_{-}}{G_{-}^2}}I_{\Rbrack 0,T\wedge\tau\Rbrack}\is\widehat m
 +{\frac{\xi^{(G)}G- M^{(B)}+{\overline D}^{o,\mathbb F}}{G}}I_{\Rbrack 0,R\Lbrack}I_{\Rbrack 0,T\Rbrack}\is N^{\mathbb G}\nonumber\\
&\quad +\left(\E[G_T\ \big|\ {\cal G}_{\tau}]-\xi^{(G)}_{\tau}\right) I_{\Lbrack\tau, +\infty\Lbrack},
\end{align}
where
\begin{equation}\label{xiGandM(G)}
M^{(B)}_t:=\E\left[{\overline D}^{o,\mathbb F}_{\infty}-{\overline D}^{o,\mathbb F}_0\ \big|\ {\cal F}_{t\wedge T}\right],
\quad \xi^{(G)}:={\frac{d{\overline D}^{o,\mathbb F}}{d{D}^{o,\mathbb F}}},\quad {\overline D}^{o,\mathbb F}:=
\left(G_T I_{\Lbrack\tau, +\infty\Lbrack}\right)^{o,\mathbb F}.
\end{equation}
\end{theorem}

Throughout the rest of this section, we consider the following notation. Thanks to the GKW-decomposition, with respect to $S$,  
 of $G(T)$ and $M^{(B)}$ defined in (\ref{G(T)process}) and (\ref{xiGandM(G)}) respectively, we get
\begin{equation}\label{GKW4G(T)andM(B)}
G(T)=G_0(T)+\varphi^{(E)}\is S+L^{(E)},\quad M^{(B)}=M^{(B)}_0+\varphi^{(B)}\is S+L^{(B)}.
\end{equation}

The superscripts $E$ and $B$ in the strategies $\varphi^{(\cdot,\mathbb H)}$  and the remaining risks $L^{(\cdot,\mathbb H)}$
refer to the type of contract (i.e., the letter ``$E$" refers to the pure endowment insurance
contract,
 while the letter ``$B$" refers to the longevity bond). Then, throughout this section, we consider
 \begin{align}
& \varphi^{(E,\mathbb G)}:=\varphi^{(E)}\left(G_{-}+\varphi^{(m)}\right)^{-1}I_{\Rbrack 0,\tau\Rbrack}, \quad L^{(E,\mathbb G)}:=
G_{-}^{-1}I_{\Rbrack 0,\tau\Rbrack}\is\widehat {L^{(E)}}
-\frac{\varphi^{(E,\mathbb G)}}{G_{-}}\is\widehat{L^{(m)}}+\label{Phi(E,G)} \\
& \hspace*{8cm} -G_{-}(T)G_{-}^{-2}I_{\Rbrack 0,\tau\Rbrack}\is\widehat{m} -G(T)G^{-1}I_{\Lbrack 0,R\Lbrack}\is \left( N^{\mathbb G}\right)^T,\nonumber\\
& \varphi^{(B,\mathbb G)}:=\varphi^{(B)}\left(G_{-}+\varphi^{(m)}\right)^{-1}I_{\Rbrack 0,\tau\Rbrack}, \quad L^{(B,\mathbb G)}:=
{{I_{\Rbrack 0,\tau\Rbrack}}\over{G_{-}}}\is\widehat{L^{(B)}}-
{{\varphi^{(B,\mathbb G)}}\over{G_{-}}}\is\widehat{L^{(m)}}+L^{(1)}, \label{Phi(B,G)}\\
& L^{(1)}:=\left(-M^{(B)}_{-}+{\overline D}^{o,\mathbb F}_{-}\right)G_{-}^{-2}I_{\Rbrack 0,T\wedge\tau\Rbrack}\is\widehat m
+ \left[\xi^{(G)}+\left(- M^{(B)}+{\overline D}^{o,\mathbb F}\right)G^{-1}I_{\Lbrack 0,R\Lbrack}\right]\is\left( N^{\mathbb G}\right)^T\label{L1}\\
& \hspace*{6.75cm} 
+ \E\left[G_T-\xi^{(G)}_{\tau}\ \big|\ {\cal G}_{\tau}\right]I_{\Lbrack\tau, +\infty\Lbrack}.\nonumber
 \end{align}
Here $\xi^{(G)}$ and $\overline{D}^{o,\mathbb F}$ are given by (\ref{xiGandM(G)}). Now, we are in the stage of announcing our main result of this section.

\begin{theorem}\label{hedgingEG} Suppose that (\ref{mainassumpOn(X,tau)}) holds, and let $h\in L^2({\cal O}(\mathbb F), P\otimes D)$.
Consider $(\varphi^{(B,\mathbb G)}, L^{(B,\mathbb G)})$ and $(\varphi^{(E,\mathbb G)}, L^{(E,\mathbb G)})$
 defined in (\ref{GKW4G(T)andM(B)})-(\ref{Phi(E,G)}) and (\ref{GKW4G(T)andM(B)})-(\ref{Phi(B,G)}) respectively, and
$(\xi^{(h,\mathbb G)}, L^{(h,\mathbb G)})$ given by (\ref{riskminG})-(\ref{riskminGremaining}). Then the following assertions hold.\\
{\rm{(a)}} Consider the market model $\left(S^{\tau}, B^{\tau},\mathbb G\right)$.
Then the risk-minimizing strategy and the remaining risk in this market model, for the insurance contract with
 payoff $h_{\tau}$, are denoted by $(\xi^{(h,1)}, \xi^{(h,2)})$ and $L^{(\mathbb G)}$ respectively, satisfy
\begin{equation*}
H:=\ ^{o,\mathbb G}(h_{\tau})=H_0+\xi^{(h,1)}\is S^{\tau}+\xi^{(h,2)}\is B^{\tau}+L^{(\mathbb G)},
\end{equation*} 
and are given by 
\begin{equation}
\xi^{(h,2)}:=\frac{d\langle L^{(h,\mathbb G)},L^{(B,\mathbb G)}\rangle^{\mathbb G}}{d\langle L^{(B,\mathbb G)}\rangle^{\mathbb G}},
\quad \xi^{(h,1)}:=\xi^{(h,\mathbb G)}-\varphi^{(B,G)}\xi^{(h,2)},\quad
L^{(\mathbb G)}:=L^{(h,\mathbb G)}-\xi^{(h,2)}\is L^{(B,\mathbb G)}.\
\end{equation}
{\rm{(b)}} Consider the market model $\left(S^{\tau}, P^{(1)},\mathbb G\right)$.
Then the risk-minimizing strategy and the remaining risk in this market model, for the insurance contract with payoff
 $h_{\tau}$, are denoted by $(\widetilde{\xi}^{(h,1)}, \widetilde{\xi}^{(h,2)})$ and $\widetilde{L}^{(\mathbb G)}$ respectively, satisfy
\begin{equation*}
H:=\ ^{o,\mathbb G}(h_{\tau})=H_0+\widetilde{\xi}^{(h,1)}\is S^{\tau}+\widetilde{\xi}^{(h,2)}\is P^{(1)}+\widetilde{L}^{(\mathbb G)},
\end{equation*}
and are given by
\begin{equation}
\widetilde{\xi}^{(h,2)}:={{d\langle L^{(h,\mathbb G)},L^{(E,\mathbb G)}\rangle^{\mathbb G}}\over{d\langle
 L^{(E,\mathbb G)}\rangle^{\mathbb G}}}, \quad
\widetilde{\xi}^{(h,1)}:=\xi^{(h,\mathbb G)}-\varphi^{(E,G)}\widetilde{\xi}^{(h,2)},\quad  \widetilde{L}^{(\mathbb G)}:=L^{(h,\mathbb G)}-\widetilde{\xi}^{(h,2)}\is L^{(E,\mathbb G)}.
\ \
\end{equation}
{\rm{(c)}} Consider the market model $\left(S^{\tau}, P^{(1)}, B^{\tau},\mathbb G\right)$.
Then the risk-minimizing strategy and the remaining risk in this market, for the insurance contract with payoff $h_{\tau}$,
 are denoted by $(\overline{\xi}^{(h,1)}, \overline{\xi}^{(h,2)}, \overline{\xi}^{(h,3)})$ and $\overline{L}^{(\mathbb G)}$ respectively, and satisfy
\begin{equation*}
H:=H_0+\overline{\xi}^{(h,1)}\is S^{\tau}+\overline{\xi}^{(h,2)}\is P^{(1)}+\overline{\xi}^{(h,3)}\is B^{\tau}+\overline{L}^{(\mathbb G)},
\end{equation*}
where
\begin{align*}
\overline{\xi}^{(h,2)}&:={{\widetilde{\xi}^{(h,2)}
	-\psi^{(E,B)}{\xi}^{(h,2)}}\over{1-\psi^{(E,B)}\theta^{(E,B)}}}I_{\{\psi^{(E,B)}\theta^{(E,B)}\not=1\}}, \ \ \ \ \ 
\overline{\xi}^{(h,3)}:={{\xi^{(h,2)}-\theta^{(E,B)} \widetilde{\xi}^{(h,2)}}\over{1-\psi^{(E,B)}\theta^{(E,B)}}}I_{\{\psi^{(E,B)}\theta^{(E,B)}\not=1\}}, \\
\overline{\xi}^{(h,1)}&:=\xi^{(h,\mathbb G)}-\varphi^{(E,G)}\overline{\xi}^{(h,2)}-\varphi^{(B,G)}\overline{\xi}^{(h,3)}\ \ \hskip 9mm
\overline{L}^{(\mathbb G)}:= L^{(h,\mathbb G)}-\overline{\xi}^{(h,2)}\is L^{(E,\mathbb G)}-\overline{\xi}^{(h,3)}\is L^{(B,\mathbb G)}.
\end{align*}
Here $\theta^{(E,B)}$ and $\psi^{(E,B)}$ are given by

$$
\theta^{(E,B)}:={{d\langle L^{(E,\mathbb G)},L^{(B,\mathbb G)}\rangle^{\mathbb G}}\over{d\langle
L^{(B,\mathbb G)}\rangle^{\mathbb G}}},\ \ \ \psi^{(E,B)}:={{d\langle L^{(E,\mathbb G)},L^{(B,\mathbb G)}\rangle^{\mathbb G}}\over{d\langle
L^{(E,\mathbb G)}\rangle^{\mathbb G}}}.$$
\end{theorem}

\begin{proof} This proof is achieved in three parts where we prove assertions (a), (b) and (c) respectively. \\
 {\bf Part 1):} By combining (\ref{GKW4G(T)andM(B)}) and (\ref{decompositionXhat}) in \eqref{LongevityBond}, we derive
\begin{equation}\label{equa300}
B^{\tau}=B_0+\varphi^{(B,\mathbb G)}\is S^{\tau}-{{\varphi^{(B,\mathbb G)}}\over{G_{-}}}\is\widehat{L^{(m)}}
+{{I_{\Rbrack 0,\tau\Rbrack}}\over{G_{-}}}\is\widehat{L^{(B)}}+L^{(1)}=
B_0+\varphi^{(B,\mathbb G)}\is S^{\tau}+L^{(B,\mathbb G)},
\end{equation}
where $\varphi^{(B,\mathbb G)}$ and $L^{(1)}$ are given in (\ref{Phi(B,G)}). Then, by inserting this equality in
$H=H_0+\xi^{(h,1)}\is S^{\tau}+\xi^{(h,2)}\is B^{\tau}+L^{(\mathbb G)}$, we obtain
$$
H=H_0+\left[\xi^{(h,1)}+\varphi^{(B,\mathbb G)}\xi^{(h,2)}\right]\is S^{\tau}+\xi^{(h,2)}\is L^{(B,\mathbb G)}+L^{(\mathbb G)}.$$
Thus, by comparing this resulting equation with
\begin{equation}\label{decompo4H}
H=H_0+\xi^{(h,\mathbb G)}\is S^{\tau}+L^{(h,\mathbb G)},
\end{equation}
where $\xi^{(h,\mathbb G)}$ and $L^{(h,\mathbb G)}$ are given by (\ref{riskminG})-(\ref{riskminGremaining}), we conclude that
$$
\xi^{(h,\mathbb G)}=\xi^{(h,1)}+\varphi^{(B,\mathbb G)}\xi^{(h,2)},\ \ L^{(h,\mathbb G)}= \xi^{(h,2)}\is L^{(B,\mathbb G)}+L^{(\mathbb G)}.$$
This is due to the fact that $L^{(\mathbb G)}$ is orthogonal to $(S^{\tau},B^{\tau})$ if and only if it is also orthogonal to
 $(S^{\tau}, L^{(B,\mathbb G)})$. Therefore, the proof of assertion (a) follows immediately.\\
 {\bf Part 2):}  To prove assertion (b),  similarly we derive the following decomposition for $P^{(1)}$ in \eqref{PureEndowment}
\begin{align}
P^{(1)}&= P^{(1)}_0+\varphi^{(E,\mathbb G)}\is S^{\tau}-{{\varphi^{(E,\mathbb G)}}\over{G_{-}}}\is\widehat{L^{(m)}}+
{{I_{\Rbrack 0,\tau\Rbrack}}\over{G_{-}}}\is\widehat{L^{(E)}}
-{{G(T)_{-}}\over{G_{-}^2}}I_{\Rbrack 0,\tau\Rbrack}\is\widehat m
-{{ G(T)}\over{G}}I_{\Lbrack 0,R\Lbrack}\is\left( N^{\mathbb G}\right)^T\nonumber\\
&= P^{(1)}_0+\varphi^{(E,\mathbb G)}\is S^{\tau}+L^{(E,\mathbb G)}.\label{equa333}
\end{align}
Then, by combining this with (\ref{decompo4H}), the proof of assertion (b) follows immediately.\\
 {\bf Part 3):}  Herein, we prove assertion (c). By inserting (\ref{equa300}) and (\ref{equa333}) in
 $H=H_0+\overline{\xi}^{(h,1)}\is S^{\tau}+\overline{\xi}^{(h,2)}\is P^{(1)}+\overline{\xi}^{(h,3)}\is B^{\tau}+L^{(\mathbb G)},$  we obtain
$$
H=H_0+\left[\overline{\xi}^{(h,1)}+\varphi^{(E,\mathbb G)}\overline{\xi}^{(h,2)}+\varphi^{(B,\mathbb G)}\overline{\xi}^{(h,3)}\right]\is S^{\tau}+
\overline{\xi}^{(h,2)}\is L^{(E,\mathbb G)}+\overline{\xi}^{(h,3)}\is L^{(B,\mathbb G)}+\overline{L}^{(\mathbb G)}.$$
Therefore, the proof of assertion (c) follows immediately from combining this with (\ref{decompo4H}) and the fact that
 the orthogonality of $\overline{L}^{(\mathbb G)}$ to $(S^{\tau}, P^{(1)},B^{\tau})$ is equivalent to
the orthogonality of $\overline{L}^{(\mathbb G)}$ to $(S^{\tau}, L^{(E,\mathbb G)},L^{(B,\mathbb G)})$. This ends proof of the theorem.
\end{proof}

Up to our knowledge, Theorem \ref{hedgingEG} generalizes all the existing literature on risk-minimizing strategies using mortality securitization in many directions.  Our approach in this theorem, which is based essentially on our optional martingale decomposition of Theorem \ref{TheoRepresentation} and the resulting risk decomposition of Theorem \ref{contractsStructures}, allows us to work on any model $(S,\tau)$ fulfilling (\ref{mainassumpOn(X,tau)}). As aforementioned, this assumption covers all the cases treated in the literature and goes beyond that. The reader can see easily this fact by comparing our framework to those considered in \cite{barbarin08,  biaginibotero15, biaginirheinlander16, biaginischreiber13} and the references there in to cite few. Indeed, in \cite{biaginischreiber13} the assumptions include H-hypothesis (i.e., all $\mathbb F$-local martingales are $\mathbb G$-local martingale), $\tau$ avoids the $\mathbb F$-stopping times and the hazard rate exists, and/or the mortality follows affine models. In \cite{biaginibotero15,  biaginirheinlander16}, the authors assume the independence between the stock price process and the mortality rate process, and consider the Brownian filtration. Barbarin assumes, in \cite{barbarin08}, that the mortality follows the  Heath-Jarrow-Morton model, and considers the Brownian filtration for $\mathbb F$.

 Furthermore, our results in Theorem \ref{hedgingEG}  are very explicit and more importantly they explain the impact of the securitization on the pair of risk-minimizing strategy and the remaining risk in the following sense.   For any securitization model ${\cal S}:=(S^{\tau}, Y^{(1)}, Y^{(2)}, \mathbb G)$, where $Y^{(i)}$ denotes the price process of the $i^{th}$  mortality security,  we describe in Theorem \ref{hedgingEG} very precisely how the pair of the risk-minimizing strategy and the remaining risk associated with this securitization model $(\xi^{({\cal S})}, L^{(\cal S)} )$ is obtained from the pair of the case without securitization $(\xi^{(h,\mathbb G)}, L^{(h,\mathbb G)})$, and/or from the pair that is associated with the securitization model $(S^{\tau}, Y^{(i)}, \mathbb G)$, $i=1,2$.\\
 
When $\tau$ is independent of ${\cal F}_{\infty}$ such that $P(\tau>T)>0$, then on the one hand (\ref{PureEndowment}) becomes
 \begin{equation}\label{Pg}
 P^{(g)}=P^{(g)}_0-{\frac{ gP(\tau>T)}{P(\tau>\cdot)}}\is \left( N^{\mathbb G}\right)^T=
 P^{(g)}_0-{\frac{ gP(\tau>T)}{P(\tau>\cdot)}}\is\left({\overline N}^{\mathbb G}\right)^T,
 \end{equation}
 where $\overline{N}^{\mathbb G}:= D-G_{-}^{-1}I_{\Rbrack 0, \tau\Rbrack}\is D^{p,\mathbb F}$ is the $\mathbb G$-martingale in the Doob-Meyer decomposition of $D$.
 
 On the other hand, the longevity bond has a constant price process equal to $G_T$, and hence it cannot be used for hedging any risk! Thus,  under the independence condition between $\tau$ and $\mathbb F$, the pure endowment insurance with benefit one (the contract that pays one dollars to the beneficiary if s/he survives) is more adequate to hedge pure mortality/longevity risk in insurance
 liabilities, while the longevity bond has no effect at all.\\

\begin{appendices}

 %%%%%%%%%%%%%%%%%%%%%%%%%%%%%%%%%%%%%%%%%%%%%%%%%%%%%%%%%%%%%%%%%

\section{Proof of Lemma \ref{consequences4MainAssum}}\label{AppendixproofLemma}

 The proof of this lemma requires two lemmas that we start with.

\begin{lemma}\label{GFCompensators} Let $V$ be an $\mathbb F$-adapted process with $\mathbb F$-locally integrable variation. Then we have
\begin{equation}\label{GcompensatorofVbeforetaugeneral}
\comg{V^{\tau}} =(G_{-})^{-1}I_{\Rbrack 0,\tau\Rbrack}\is\bigl({\widetilde G}\is V\bigr)^{p,\mathbb F}  . \end{equation}
\end{lemma}
 For the proof of this lemma, we refer the reader to \cite[Lemma 3.1]{aksamit2017}. 
 \begin{lemma} \label{lemmaapp}
For a non-negative \(\mathbb H\)-optional process, \(\phi\), such that \(0\le \phi \le 1\) and \(V \in {\cal A}^+_{\loc} (\mathbb H)\),
the following assertions hold. \\
(i) There exists an \(\mathbb H\)-predictable process, \(\psi\), satisfying
\[
0 \le \psi \le 1 \quad \text{ and } \quad \big(\phi \is V\big)^{p,\mathbb H} = \psi \is V^{p,\mathbb H}.
\]
(ii)
If \(P\otimes V(\{\phi =0 \}) = 0\), then \(\psi\) can be chosen strictly positive for all \((\omega, t) \in \Omega \times \mathbb R_+\).
\end{lemma}
\begin{proof}
(i) Since \(\phi\le1\), it is clear that $d (\phi \is V)^{p,\mathbb H} \ll d V^{p,\mathbb H}$, $P$-a.s.. Hence, there
 exists a non-negative and \(\mathbb H\)-predictable process \(\psi^{(1)}\) such that
\begin{equation} \label{equalitypsi1}
\big(\phi \is V\big)^{p,\mathbb H} = \psi^{(1)} \is V^{p,\mathbb H}.
\end{equation}
As a result, we derive $
0 = I_{\{\psi^{(1)}>1\}} \is\Big[ (\phi \is V)^{p,\mathbb H} - \psi^{(1)} \is V^{p,\mathbb H}\Big]
 = \Big( (\phi - \psi^{(1)} ) I_{\{\psi^{(1)}>1\}}\is V  \Big)^{p,\mathbb H},$
 and deduce that $P\otimes V^{p,\mathbb H}(\{\psi^{(1)}>1\})=0$. Thus, by putting $\psi=\psi^{(1)}\wedge 1$, assertion (a) follows. \\
(ii) It is clear from (\ref{equalitypsi1}) that $0=I_{\{\psi^{(1)}=0\}}\is (\phi\is V)^{p,\mathbb H}=
(\phi I_{\{\psi^{(1)}=0\}}\is V)^{p,\mathbb H}$. This implies that $\{\psi^{(1)}=0\} \subset \{\phi=0\}$ $P\otimes V$-a.e..
 Therefore, assertion (b) follows from  putting $\psi=\psi^{(1)}\wedge 1 +I_{\{\psi^{(1)}=0\}}$, and the proof of the lemma is completed.
\end{proof}

\begin{proof}[Proof of Lemma  \ref{consequences4MainAssum}]:  This proof consists of three steps, where we prove assertions (a), (b) and (c)-(d) respectively. \\
{\bf Step 1:} Thanks to \cite{jeulin80}, $S^{\tau}-G_{-}^{-1}I_{\Rbrack 0,\tau\Rbrack}\is\langle S,m\rangle^{\mathbb F}$
is a $\mathbb G$-local martingale. Thus, by combining this with the second assumption in (\ref{mainassumpOn(X,tau)}) (i.e., $\langle S,m\rangle^{\mathbb F}\equiv 0$), we deduce that $S^{\tau}$
is $\mathbb G$-local martingale. Thus, the assertion (a) follows immediately.\\
{\bf Step 2:} Due to the third assumption in (\ref{mainassumpOn(X,tau)}), it holds that \(\Delta S I_{\Lbrack\widetilde R\Rbrack}=\Delta S I_{\{\widetilde G=0<G_{-}\}} \equiv 0\). Thus, for any
$L\in {\cal M}_{loc}(\mathbb F)$ orthogonal to $S$, we have
	$$
	[\widehat{L},S^\tau] = G_{-}{\widetilde{G}}^{-1}I_{\Rbrack 0,\tau\Rbrack} \is [L,S]+
\ ^{p,\mathbb F}(\Delta LI_{\Lbrack\widetilde R\Rbrack} )\is S^{\tau}.
	$$
	Since $\ ^{p,\mathbb F}(\Delta LI_{\Lbrack\widetilde R\Rbrack} )\is S^{\tau}$ a $\mathbb G$-local martingale and $\Delta L\Delta S I_{\Lbrack\widetilde R\Rbrack} \equiv 0$, Lemma \ref{GFCompensators}% (given at the end of this proof) 
	 implies that
	$$\langle\widehat{L},S^\tau\rangle^{\mathbb G}=I_{\Lbrack 0,\tau\Rbrack} \is\langle L,S\rangle^{\mathbb F}\equiv 0.$$
	This proves assertion (b).\\
{\bf Step 3:} Since $m$ is bounded and orthogonal to $S\in {\cal M}_{loc}^2(\mathbb F)$, it is clear that $U:=I_{\{G_{-}>0\}}\is [S,m]\in {\cal M}_{0,loc}^2(\mathbb F)$. Then, an application
of the Galtchouk-Kunita-Watanabe decomposition of $U$ with respect to $S$, we get the first property in (\ref{decompositionUandhatU}).
To prove the second property in  (\ref{decompositionUandhatU}), we remark that  \( [U, S] = \Delta m I_{\{G_{-}>0\}} \is [S,S]\),
and put
	\[
	W := G_- \is [S,S] + [U,S] = \widetilde{G} I_{\{G_{-}>0\}} \is [S,S]\ \ \ \ \ \mbox{and}\ \ \ \ V := I_{\{G_{-}>0\}} \is [S,S].
	\]
	A direct application of Lemma \ref{lemmaapp} %, see the second lemma at the end this proof, 
	to the pair $(V,\widetilde G+I_{\{\widetilde G=0\}})$ (it is easy to see that
the assumptions of this lemma are fulfilled as \(P\otimes V(\{\phi =0 \}) = P\otimes I_{\{G_{-}>0\}} \is [S,S](\{\widetilde{G} =0 \}) = 0\)
which follows from \( I_{\{\widetilde{G} =0 < G_{-} \}} \Delta S =0 \)), we deduce that the existence of $\mathbb F$-predictable \(\psi\) such that \(0< \psi \le 1\) and
	\[
	W^{p,\mathbb F} = \psi I_{\{G_{-}>0\}} \is\langle S,S \rangle^{\mathbb F} =
 (G_- + \varphi^{(m)}) I_{\{G_->0\}} \is\langle S,S \rangle^{\mathbb F}.
	\]
	This completes the proof of assertion (c).  Thus, the rest of the proof focuses on proving assertion (d). To this end, we notice that due to $\Delta U I_{\{\widetilde G=0<G_{-}\}}=-G_{-}\Delta S  I_{\{\widetilde G=0<G_{-}\}}=0$, it is clear that
	$$
	\widehat{U} = U^\tau - {\widetilde{G}}^{-1} I_{\Lbrack 0,\tau\Rbrack} \is [U,m] =  I_{\Lbrack 0,\tau\Rbrack} \is [S,m] -
{\widetilde{G}}^{-1} I_{\Lbrack 0,\tau\Rbrack} \Delta m \is [S,m] = I_{\Rbrack 0,\tau\Rbrack} 	G_- {\widetilde{G}}^{-1} \is U.
	$$
	As a result, on the one hand, we get
	\begin{equation} \label{hatXinhatU}
	\widehat{S}  =   S^\tau -
 G_-^{-1} I_{\Rbrack 0,\tau\Rbrack}\is\widehat{U}.
	\end{equation}
	On the other hand, due to (\ref{decompositionUandhatU}), we derive
	$$
	\widehat{U} =  \varphi^{(m)} \is\widehat{S} + \widehat{L^{(m)}} =
  \varphi^{(m)} \is S^\tau - \varphi^{(m)} G_-^{-1} I_{\Rbrack 0,\tau\Rbrack} \is\widehat{U}  +\widehat{L^{(m)}} .
	$$
	Solving for \(\widehat{U}\), we get
	\[
	(G_- + \varphi^{(m)})G_-^{-1} I_{\Rbrack 0,\tau\Rbrack} \is\widehat{U}   =
 \varphi^{(m)} \is S^\tau + \widehat{L^{(m)}}.
	\]
	By inserting this equality in (\ref{hatXinhatU}), (\ref{decompositionXhat}) follows immediately, and the proof of
the lemma is complete.\end{proof} 
\end{appendices}

\end{document}